\documentclass[manuscript,nonacm]{acmart}

\usepackage{amsthm,amsmath}
\usepackage{graphicx}
\usepackage{color}
\usepackage{hyperref}
\usepackage[ruled,vlined,linesnumbered]{algorithm2e}
\usepackage{todonotes}
\usepackage{nicefrac}
\usepackage{wrapfig}

\usepackage{mathtools}
\DeclarePairedDelimiter\ceil{\lceil}{\rceil}

\newcommand{\m}[1]{\mathcal{#1}}
\newcommand{\cl}{\mbox{\footnotesize \sf CL}}
\newcommand{\D}{{\Delta}}
\newcommand{\ID}{\mbox{\sc{id}}}
\newcommand{\id}{\mbox{\footnotesize $\ID$}}
\newcommand{\collect}{\textsf{collect}}
\newcommand{\proj}{\mathsf{proj}}
\newcommand{\TS}{{\normalfont\textsf{test\&set}}}

\newtheorem{definition}{Definition}
\newtheorem{lemma}{Lemma}

\newtheorem{theorem}{Theorem}
\newtheorem{corollary}{Corollary}

\newtheorem{claim}{Claim}

\begin{document}

\title[A Speedup Theorem for Asynchronous Computation]{A Speedup Theorem for Asynchronous Computation with Applications to Consensus and Approximate Agreement}

\author{Pierre Fraigniaud}
\thanks{Pierre Fraigniaud received support from the ANR project DUCAT.}
\email{pierre.fraigniaud@irif.fr}
\affiliation{%
	\institution{IRIF,  CNRS and Universit\'e Paris Cit\'e}
	\city{Paris}
	\country{France}
}

\author{Ami Paz}
\email{ami.paz@lisn.fr}
\affiliation{%
	\institution{LISN,  CNRS and Universit\'e Paris Saclay}
	\city{Paris}
	\country{France}
}

\author{Sergio Rajsbaum}
\email{rajsbaum@im.unam.mx}
\thanks{Part of this work was done while Sergio Rajsbaum was visiting LIX at \'Ecole Polytechnique, and IRIF at Universit\'e Paris Cit\'e, France. He received additional support from UNAM-PAPIIT IN106520}
\affiliation{%
	\institution{Instituto de Matem\'aticas, UNAM}
	\city{Mexico City}         
	\country{Mexico}   
}

\begin{abstract}

We study two fundamental problems of distributed computing,  \emph{consensus} and  \emph{approximate agreement},  through a novel approach for proving lower bounds and impossibility results, that we call the \emph{asynchronous speedup theorem}. 
For a given $n$-process task 
$\Pi$ and a given computational model $M$, we define a new task, called the \emph{closure} of~$\Pi$ with respect to~$M$. 
The asynchronous speedup theorem states that
if  a task $\Pi$ is solvable in $t\geq 1$ rounds in~$M$, then its closure w.r.t.~$M$ is solvable in $t-1$ rounds in~$M$.
We prove this theorem for iterated models, as long as the model allows solo executions. 
We illustrate the power of our asynchronous speedup theorem by providing a new proof of the wait-free impossibility of consensus using read/write registers, 
and a new proof of the wait-free impossibility of solving consensus using registers and test\&set objects for $n>2$. 
The proof is merely by showing that, in each case, the closure of consensus (w.r.t. the corresponding model) is consensus itself.
Our main application is the study of the power of  additional objects, namely test\&set and binary consensus, for wait-free solving approximate agreement \emph{faster}.
By analyzing the closure of approximate agreement w.r.t. each of the two models, we show that while these objects are more powerful than read/write registers from the computability perspective,
they are not more powerful as far as helping solving approximate agreement faster is concerned. 
\end{abstract}

\maketitle

\section{Introduction}

One of the  goals of the theory of distributed computing is to understand what can be achieved
by a given distributed system. A great deal of research has been devoted to develop techniques for showing that certain problems cannot be solved, or to prove lower bounds on the resources needed to solve some given problem~\cite{AttiyaEllenBook,FichR03hundred,Lynch89hundred}. 
Much of this research is motivated by the need to understand the solvability of \emph{consensus} and \emph{approximate agreement}, 
two of the central tasks of distributed computing. 
In particular, this line of research lead to the introduction of bivalency~\cite{FLP85} and indistinguishability~\cite{AR20indist} arguments, and later on to the discovery of the close connection between distributed computing and algebraic topology, based on the study of topological invariants of the protocol complex representing a given model of computation~\cite{bookHerlihyKR2013}.
In this paper we propose a new  approach to prove unsolvability results and time lower bounds.

\subsection{Asynchronous Speedup Theorem}

We introduce the \emph{asynchronous speedup theorem} for shared-memory computing, 
inspired by the recent speedup technique~\cite{Balliu0HORS19,Brandt19} for the \textsf{LOCAL} model of synchronous failure-free distributed computing in networks.
Specifically, we consider the iterated shared-memory model, where the memory is organized in arrays $M_r$, $r\geq 1$, of $n$ single-writer/multiple-readers (SWMR) registers (one per process), and each round $r$ of the protocol is executed on the array~$M_r$.
Recall that a task $\Pi=(\m{I},\m{O},\Delta)$ is defined by a collection~$\m{I}$ of possible input configurations, a collection~$\m{O}$ of possible output configurations, and a map $\Delta:\m{I}\to 2^{\m{O}}$ specifying an  input-output relation. 
A subset of processes starting with inputs as specified by a configuration $\sigma\in \m{I}$ must output values that define a configuration $\tau\in\Delta (\sigma)$. 

For our speedup theorem, we define the \emph{closure} of a task~$\Pi$ with respect to a computational model~$M$,  denoted $\cl_M(\Pi)=(\m{I},\m{O}',\Delta')$, which
is supposed to be a slightly easier version of~$\Pi$ under~$M$. 
Recall that an algorithm solves a task $\Pi$  in $t$~rounds if, whenever a subset of processes 
start with inputs as specified by a configuration $\sigma\in \m{I}$, after executing $t$ rounds, they must output values that define a configuration $\tau\in\Delta (\sigma)$. 
Our asynchronous speedup theorem states that, for every $t\geq 1$, if a task $\Pi$ is solvable in $t$ rounds in~$M$, then $\cl_M(\Pi)$ is solvable in $t-1$ rounds in~$M$. 
Therefore, to establish a lower bound on the number~$t$ of rounds for solving~$\Pi$, it suffice to show that
iterating the closure operation for $t$ times starting with~$\Pi$, results in a task that is not solvable in zero rounds. 
Alternately, impossibility results follow when the closure operator can be performed infinitely many times without ever reaching a task solvable in zero rounds, which is typically the case when the closure of a task is the task itself.

Notice that the notion of ``$\cl_M(\Pi)$ being easier to solve than~$\Pi$'' is with respect to the model~$M$ of computation.
We consider models obtained by extending the standard shared-memory model with objects more powerful than SWMR registers---the task $\cl_M(\Pi)$ then depends on the nature of these objects.
Intuitively, $\cl_M(\Pi)$ is supposed to be easier than $\Pi$ because it allows, on each given input configuration $\sigma\in\m{I}$ to output some  values that are illegal according to~$\Delta$.
To prevent $\cl_M(\Pi)$ from being \emph{too} easy in~$M$, it is however required that such an output configuration $\tau$ is ``close to a legal solution'', in the following sense: there must exists a $1$-round algorithm which, starting with the configuration $\tau$ as input, computes legal outputs for $\sigma$, that is, outputs a configuration in~$\Delta(\sigma)$.
If there is such a $1$-round algorithm for~$M$, then the configuration $\tau$ is added to the collection $\Delta'(\sigma)$ of legal output configurations
for $\sigma$ in the closure task~$\cl_M(\sigma)$.

\subsection{Our Results}

We  present our asynchronous speedup theorem for any iterated asynchronous model with SWMR registers, as long as this model allows \emph{solo} executions (Theorem~\ref{thm:speedup}). Such models include the usual 
{wait-free} iterated immediate snapshot (IIS), iterated snapshot, and iterated collect models. 
They also include  \emph{affine models}~\cite{KRH18} that allow solo executions, such
as the  $k$-concurrency model~\cite{genUniver2011}. Roughly, an affine model is obtained from the IIS model by
removing some executions. Iterated models with solo executions that \emph{add} some executions have
also been studied, such as the $d$-solo models~\cite{soloModels17}, where $d$ processes may run solo in the same execution. 
Note that lower bounds for iterated models, as the ones we prove here, immediately apply to the non-iterated variants of the same models---in a non-iterated model, the adversary can still choose to have only executions where all processes finish their $r$-th step before taking their $(r+1)$-th step.

We illustrate the use of our asynchronous speedup theorem by revisiting the renowned wait-free impossibility of solving consensus.
We show how to derive this impossibility result  in a simple
way from our speedup theorem  (Corollary~\ref{cor:impossibility-consensus}):
to prove the impossibility of consensus, it is enough to establish that the closure of consensus is consensus itself.
Next, we present an extension of our speedup theorem for models as above augmented with objects more powerful than SWMR registers (Theorem~\ref{thm:speedup-extended}).
We illustrate its  usefulness  with a new proof of the impossibility for consensus among $n\geq 3$ processes when  the processes have access to a  \TS\/ object at each round (Corollary~\ref{cor: consensus impossibility with ts}). 

After considering solvability, we move 
our attention to the question of how much objects in
Herlihy's consensus hierarchy can help solving a problem \emph{faster}.
We present two examples of objects that are more powerful than read/write registers,
namely \TS\/ and binary consensus, 
and yet not more powerful as far as helping solving approximate agreement faster is concerned. 
Our proofs, based on our asynchronous speedup theorem, 
provide novel information about the topology of such extended models.
More specifically, in the second part of the paper, we illustrate the power of the asynchronous speedup theorem by establishing new lower bounds for $\epsilon$-approximate agreement, in the wait-free IIS model with access to \TS\/ or binary consensus objects. 
We first show that the number of rounds required for solving $\epsilon$-approximate agreement in the wait-free IIS model among $n\geq 3$ processes is at least $\ceil{\log_2(1/\epsilon)}$ (Corollary~\ref{lem:approx-agree-without-BB}),
which is tight, 
and then we show that augmenting the model with \TS{} objects does not reduce the time complexity at all (Theorem~\ref{theo:approx-agree-with-ts}),
assuming that each of the processes calls all the \TS{} objects.
These results are obtained by simply showing that the closure of $\epsilon$-approximate agreement in the considered model is the $(2\epsilon)$-approximate agreement task. 
For $n=2$, $\epsilon$-approximate agreement can be solved in a single round in the wait-free IIS model with \TS,
while in the wait-free IIS model without additional objects
we show that $\epsilon$-approximate agreement among $n=2$ processes requires $\lceil\log_3 1/\epsilon\rceil$ rounds
(Corollary~\ref{lem:approx-agree-without-BB}).
This bound is proved by showing that the closure of $\epsilon$-approximate agreement in this case is the $(3\epsilon)$-approximate agreement task,
and is also known to be tight.

Finally, we show that the number of rounds required by $n\geq 3$ processes
for solving the $\epsilon$-approximate agreement task in the wait-free IIS model augmented with a binary consensus object
is at least $\min\{\ceil{\log_2 1/\epsilon},\ceil{\log_2 n}-1\}$
(Theorem~\ref{thm:approx-agree-with-bc}),
when we assume that each input to the binary consensus object depends solely on the process~ID and of the round number. 
In this context, the closure of  $\epsilon$-approximate agreement  is not necessarily $(2\epsilon)$-approximate agreement. However, we can show that, for a set of at least half of the processes, the closure of  $\epsilon$-approximate agreement when only the processes in this set participate is $(2\epsilon)$-approximate agreement. By iterating this argument, we get that if there is a $t$-round algorithm for solving $\epsilon$-approximate agreement  among $n$ processes, then there is a $t-r$ round algorithm for solving $(2^r\epsilon)$-approximate agreement among $n/2^r$ processes. We lose an additive factor $-1$ for technical reasons related to 
the fact that $\epsilon$-approximate agreement is solvable when there are only $n=2$ processes.
Note that our lower bound is essentially tight as there exists an algorithm for solving $\epsilon$-approximate agreement in $\ceil{\log_2 1/\epsilon}$ 
rounds in the wait-free IIS model~\cite{AspnesH90}, and there exists an algorithm for solving multi-value consensus in~$\ceil{\log_2 n}$ 
rounds~\cite{MRTronel2000,Raynal18bookFT} in the wait-free IIS model  with  binary consensus objects.

\subsection{Related Work}

The computability and complexity of fault-prone, asynchronous consensus and approximate agreement is a well studied topic. 
Fischer, Lynch, and Paterson (FLP)~\cite{FLP85} showed that consensus 
cannot be solved in a message passing system even if only one process may fail by halting.
Later, Herlihy~\cite{Herlihy91}
studied  wait-free computation  in the shared-memory model, and proved the consensus impossibility in this model.
Motivated by the need to analyze problems other than consensus,
Herlihy and Shavit~\cite{HerlihyS99} introduced the use of 
combinatorial topology
for reasoning about computations in asynchronous, wait-free distributed systems,
in which any number of processes may crash. 
They proved the asynchronous
computability theorem, which states a  topological condition that is necessary and sufficient for a
task to be wait-free solvable using read/write registers.
Approaches similar to the one of Herlihy and Shavit were presented in the same time by Saks and Zaharoglou~\cite{SaksZ93} and by Borowsky and Gafni~\cite{BorowskyG93},
but the formalization of Herlihy and Shavit is the one commonly used today.
To conclude, the two techniques that have been extensively used 
for proving impossibility results for consensus and approximate agreement are based on either  valency~\cite{FLP85} or connectivity  analysis~\cite{HerlihyS99}, in addition to indirect, reduction~\cite{Herlihy91}
or simulation techniques e.g.~\cite{CHJT05}.
Our technique of computing the closure of consensus provides a novel way for proving such impossibility results.
Regarding complexity issues, Hoest and Shavit~\cite{HoestS06} showed how the topological approach
can also be used for deriving time lower bounds
for wait-free algorithms in the IIS model, as we do.
Attiya, Castañeda, Herlihy and Paz~\cite{AttiyaCHP19} used similar techniques for proving time upper bounds, and 
Ellen, Gelashvili, and Zhu~\cite{EllenGZ18} used them for proving space lower bounds.
One aim of our paper is to investigate further the use of combinatorial topology in the context of asynchronous computing.

Soon after the celebrated FLP consensus impossibility result,  Dolev, Lynch, Pinter, Stark, and Weihl \cite{approx86} introduced the approximate agreement task, 
in which each process has a real valued input, and processes must agree on
output values within the range of inputs, that are at most $\epsilon$ apart. 
Aspnes and Herlihy~\cite{AspnesH90}
proved a lower bound of  $\ceil{\log_3 1/\epsilon}$ rounds and an upper bound of $\ceil{\log_2 1/\epsilon}$ rounds
in the wait-free SWMR model, a gap that was later closed by 
Hoest and Shavit~\cite{HoestS06}, who proved a tight bound on the number of rounds needed to solve approximate
agreement: $\ceil{\log_3 1/\epsilon}$ for $n=2$, and $\ceil{\log_2 1/\epsilon}$ for $n\geq 3$.
A similar result was recently proved in the context of dynamic networks
by F\"{u}gger, Nowak and Schwarz~\cite{FNS21},
where the reader can find a thorough discussion about the literature studying approximate agreement, and about the importance of this problem w.r.t.~applications.
Hoest and Shavit~\cite{HoestS06} prove their
lower bounds using a  global analysis
of the protocol complex, while F\"{u}gger,  Nowak and Schwarz~\cite{FNS21}
extend the notion of valency used for consensus~\cite{FLP85,MR02} for proving their results.
We establish these lower bounds in a novel way,
by computing the closure of the approximate agreement task.

In addition to proving the wait-free consensus impossibility, Herlihy~\cite{Herlihy91}
also derived a hierarchy of objects such that no object 
at one level  has a wait-free implementation in terms of objects at lower levels,
introducing a remarkable technique of reduction to a consensus protocol. 
There is plenty of work on extending the wait-free IIS model with more powerful
objects, but only from the computability perspective, e.g.~\cite{GafniR10,ImbsRV15}  using simulations, 
or using topology e.g.~\cite{HerlihyR95primer}.
We are not aware of any use of this technique for  complexity results, 
to study as we do, using strong objects to solve approximate agreement faster.

Last, but not least, a goal of this research is to understand the potential extension of  the recent speedup technique~\cite{Balliu0HORS19,Brandt19} from the (synchronous failure-free) \textsf{LOCAL} model to asynchronous models with crash-prone processes. In the \textsf{LOCAL} model, the uncertainly about how a hypothetical protocol
would solve a problem in one round, comes from unknown inputs of far away processes. In contrast,
in the asynchronous setting, the uncertainty is of a very different nature, about
what other processes have read. There have been few recent attempts to approach distributed network computing through the lens of algebraic topology~\cite{CastanedaFPRRT21,FraigniaudP20}. 
An extension of the speedup technique in~\cite{Balliu0HORS19,Brandt19} to arbitrary round-based full-information models has been proposed in~\cite{BastideF21}. This extension applies to wait-free shared-memory computing in particular, but only for $2$~processes. Our approach applies to an arbitrarily large number of processes.

\section{Model}
\label{sec:model}
We consider the standard \emph{read-write shared memory} distributed computing model~\cite{AWbook}, and we adopt a standard framework for modeling distributed computation, by using concepts borrowed from algebraic topology~\cite{bookHerlihyKR2013}. The reader unfamiliar with these notions is referred to Appendix~\ref{app:model} for more details. 

\subsection{Iterated Models}

We consider distributed systems with  $n\geq 2$ asynchronous processes, labeled by distinct integers from~1 to~$n$. Every process~$i$ initially knows its identity~$i$ as well as the total number~$n$ of processes in the system. We focus on generic round-based algorithms of the following form.
All read and write operations are atomic. Specifically, Algorithm~\ref{alg:generic-model} is for the \emph{iterated} model, in which the shared-memory is organized in arrays $M_r$, $r\geq 1$, of $n$ single-writer/multiple-readers registers (one per process), and each round $r$ of the protocol is executed on the array~$M_r$. 
At each round, each process~$i$ updates its so-called \emph{view}~$V_i$ by collecting the views of the other processes. To this end, it uses a \emph{collect} operation, which consists of reading all registers~$M_r[1..n]$ sequentially, in arbitrary order. The initial view of a process is its input, and, after $t$~rounds or write-collect, each process outputs a value which depends of its view after $t$~rounds, i.e., of all the information accumulated by the process during the execution of the algorithm. 

\begin{wrapfigure}[9]{R}{0.32\textwidth}
	\vspace*{-4ex}
	\center
	\resizebox{1\totalheight}{!}{
		\begin{minipage}{0.38\textwidth}
			\begin{algorithm}[H]
				\SetAlgoLined
				\SetKwFor{loop}{\!}{for $r=1$ to $t$}{end}
				$V_i \leftarrow x_i$\\
				\loop{}{
					$\mathbf{write}\;(i,V_i)$ in register $M_r[i]$\\	
					$V_i\leftarrow \mathbf{collect} \; M_r[1..n]$
				}
				$y_i \leftarrow f(i,V_i)$\\
				\textbf{output} $y_i$.
				\caption{\sl Code for process $i\in [n]$ with input $x_i$}
				\label{alg:generic-model}
			\end{algorithm}
		\end{minipage}
	}
\end{wrapfigure}

Two classical stronger variants of the collect instruction are considered in this paper. \emph{Snapshot} corresponds to the scenario in which the collect operation itself is atomic, that is, all registers $M_r[1..n]$ are globally read simultaneously at once, in an atomic manner~\cite{AWbook}. 
\emph{Immediate snapshot}~\cite{BorowskyG93,SaksZ93} corresponds to the scenario in which each write-snapshot sequence of operations is itself atomic, 
that is, not only all registers are read simultaneously at once, but the snapshot occurs ``immediately'' after the write operation:
a set of processes first all perform a write operation, and immediately after they all perform an atomic snapshot
(the order of reads inside the set of processes is not important, and so is the order of snapshot operations).
Our lower bound technique applies to the three wait-free models with write-collect, write-snapshot, or immediate snapshot, and we shall establish our lower bound for approximate agreement using the stronger model, i.e., iterated immediate snapshot, or IIS for short. 

\subsection{Task Solvability}

A \emph{task} for $n$ processes is a triple $\Pi=(\m{I},\m{O},\Delta)$ where $\m{I}$ and $\m{O}$ are $(n-1)$-dimensional  complexes, respectively called \emph{input} and \emph{output} complexes, and $\Delta:\m{I}\to 2^{\m{O}}$ is an input-output specification. Every simplex $\sigma=\{(i,x_i):i\in I\}$ of $\m{I}$, where $I=\ID(\sigma)$ is a non-empty subset of~$[n]$, defines a legal  input state corresponding to the scenario in which, for every $i\in I$, process~$i$ starts with input value~$x_i$. Similarly, every simplex $\tau={\{(i,y_i):i\in I\}}$ of $\m{O}$ defines a legal output state corresponding to the scenario in which, for every $i\in I$, process~$i$ outputs the value~$y_i$. The map~$\Delta$ is an input-output relation specifying, for every input state $\sigma\in\m{I}$, the set of output states~$\tau\in\m{O}$ with $\ID(\tau)=\ID(\sigma)$ that are legal with respect to~$\sigma$. That is, assuming that only the processes in~$\ID(\sigma)$ participate to the computation (the set of participating processes is not known a priori to the processes in~$\sigma$), these processes are allowed to output any simplex $\tau\in\Delta(\sigma)$. 
$\Delta(\sigma)$ can be viewed as a collection of simplexes, each with the same set of IDs as~$\sigma$, or, alternatively, as the complex induced by this set of simplexes.  
It is often assumed that~$\Delta$ is a \emph{carrier} map (that is, for every $\sigma,\sigma'\in\m{I}$, if $\sigma'\subseteq\sigma$ then $\Delta(\sigma')\subseteq\Delta(\sigma)$ as subcomplexes), but in this paper we do not enforce this requirement into the definition of tasks. 

Given a simplex~$\sigma=\{(i,x_i):i\in I\}\in \m{I}$ for some $I=\ID(\sigma)\subseteq [n]$,  one round of communication performed by the processes in~$\ID(\sigma)$ as in Algorithm~\ref{alg:generic-model} results in various possible simplices, depending on the interleaving of the different write and read operations. Such a simplex is of the form $\tau=\{(i,V_i):i\in I\}$, where $V_i=\{(j,x_j):j\in J_i\}$ is the \emph{view} of process~$i$ after one round. 
This view, or equivalently, the set~$J_i\subseteq I$, depends on the communication model~$M$, i.e., write-collect, write-snapshot, or immediate snapshot. 
These simplices induces a complex, denoted by~$\m{P}^{(1)}(\sigma)$, whose structure differs according to the three models considered on this paper. Figure~\ref{fig:3topologies} in Appendix~\ref{app:3topologies} displays an example of these three different complexes, for a 2-dimensional simplex~$\sigma$ (i.e., a system with three processes). In the case of  immediate snapshot, $\m{P}^{(1)}(\sigma)$ is the complex obtained by performing a chromatic subdivision of~$\sigma$~\cite{HerlihyS99}. That is, $\tau=\{(i,V_i):i\in I\}$ is in $\m{P}^{(1)}(\sigma)$ if and only if, for every $i,j\in I$, $j\in V_i$ or $i\in V_j$, and if $j\in V_i$ then $V_j\subseteq V_i$. We denote by $\Xi$ the  map transforming every $\sigma\in \m{I}$ into $\m{P}^{(1)}(\sigma)$, and by $\m{P}^{(1)}=\cup_{\sigma\in\m{I}}\m{P}^{(1)}(\sigma)$. The map~$\Xi$ depends on the communication model~$M$. 

The topological transformation from $\sigma$ to $\m{P}^{(1)}(\sigma)$ can be iterated, yielding the sequence $(\m{P}^{(t)})_{t\geq 0}$ where, for every simplex~${\sigma\in\m{I}}$, $\m{P}^{(0)}(\sigma)=\sigma$, and, for every $t\geq 1$, $\m{P}^{(t)}(\sigma)=\Xi(\m{P}^{(t-1)}(\sigma))$.  For every input complex~$\m{I}$, and every $t\geq 0$, the complex $\m{P}^{(t)}$ is called the \emph{protocol complex} after $t$ rounds. This complex is the union, for all simplices $\sigma\in\m{I}$ of the complex $\m{P}^{(t)}(\sigma)$ resulting from $t$~rounds starting from input~$\sigma$. Note that for any two input simplices~$\sigma=\{(i,x_i):i\in I\}$ and~$\sigma'=\{(i,x'_i):i\in I\}$, the complexes $\m{P}^{(1)}(\sigma)$ and $\m{P}^{(1)}(\sigma')$ are isomorphic. We denote by 
\begin{equation}\label{eq:canonicaliso}
	\chi:\m{P}^{(1)}(\sigma)\to\m{P}^{(1)}(\sigma')
\end{equation}
the canonical isomorphism that maps every vertex $v=(i,\{(j,x_j):j\in J_i\})$ of  $\m{P}^{(1)}(\sigma)$ to the vertex $\chi(v)={(i,\{(j,x'_j):j\in J_i\})}$ of $\m{P}^{(1)}(\sigma')$. Finally, recall that a task $\Pi=(\m{I},\m{O},\Delta)$ is solvable in $t$~rounds if and only if there exists a simplicial map 
$
f:\m{P}^{(t)}\to \m{O}
$
from the $t$-round protocol complex to the output complex that agrees with $\Delta$, i.e., for every $\sigma\in\m{I}$, 
$
f(\m{P}^{(t)}(\sigma))\subseteq \Delta(\sigma).
$
The mapping $f$ is defined by the vertices, i.e., 
for $\sigma=\{(i,V_i):i\in I\}$
we have $f(\sigma)=\{f(i,V_i):i\in I\}$.
The simplicial map $f$ is merely the function $f$ used in Algorithm~\ref{alg:generic-model} for computing the output values of the processes.

\section{Closure and Speedup Theorem}
\label{sec:speedup}

This section describes our main tool for deriving lower bounds and impossibility results, namely an \emph{asynchronous speedup theorem}. This theorem is based on a notion of \emph{closure} defined hereafter. 

\subsection{A Closure of a Task}

In this section, we define the notion of a \emph{closure task} (Figure~\ref{fig:closure}). As it will be shown later, given a computational model~$M$, and given a task $\Pi$, the closure of $\Pi$ w.r.t.~$M$, denoted by $\cl_M(\Pi)$, is a simpler task in the sense that if $\Pi$ is solvable in $t$ rounds in model~$M$, then $\cl_M(\Pi)$ is solvable in $t-1$ rounds in model~$M$.
In order to define the closure of a task, we first define another task, called \emph{local task} (Figure~\ref{fig:TaskDelta1}), which does not depend on~$M$, but only on~$\Pi$.
Intuitively, this task closes the one-round gap between the round-complexities of $\cl_M(\Pi)$ and $\Pi$ by taking the outputs of $\cl_M(\Pi)$ as input, and producing legal outputs for $\Pi$.

Let $\Pi=(\m{I},\m{O},\Delta)$ be a task. We say that a set $\tau\subseteq V(\m{O})$ of output vertices is \emph{chromatic} if, for any two vertices $v$ and $v'$ of~$\tau$, $\ID(v)\neq \ID(v')$. 
Intuitively, to every simplex $\sigma\in\m{I}$, and to every chromatic set of output vertices $\tau\subseteq V(\m{O})$ satisfying $\ID(\tau) =\ID(\sigma)$,
we associate a specific task with input~$\tau$, viewed as a complex, and whose objective is to construct a legal solution in~$\Delta(\sigma)$. 
Note that $\tau$ is a chromatic set of vertices in $V(\m{O})$,
but is not necessarily a simplex in~$\m{O}$. Given a chromatic complex $\m{K}$ with vertex IDs in $[n]$, and given a non-empty set $I\subseteq [n]$, we denote by $\proj_I(\m{K})$ the subcomplex of~$\m{K}$ induced by the vertices with IDs in~$I$. 

\begin{figure}[tb]
	\centering
	\includegraphics[width=9cm]{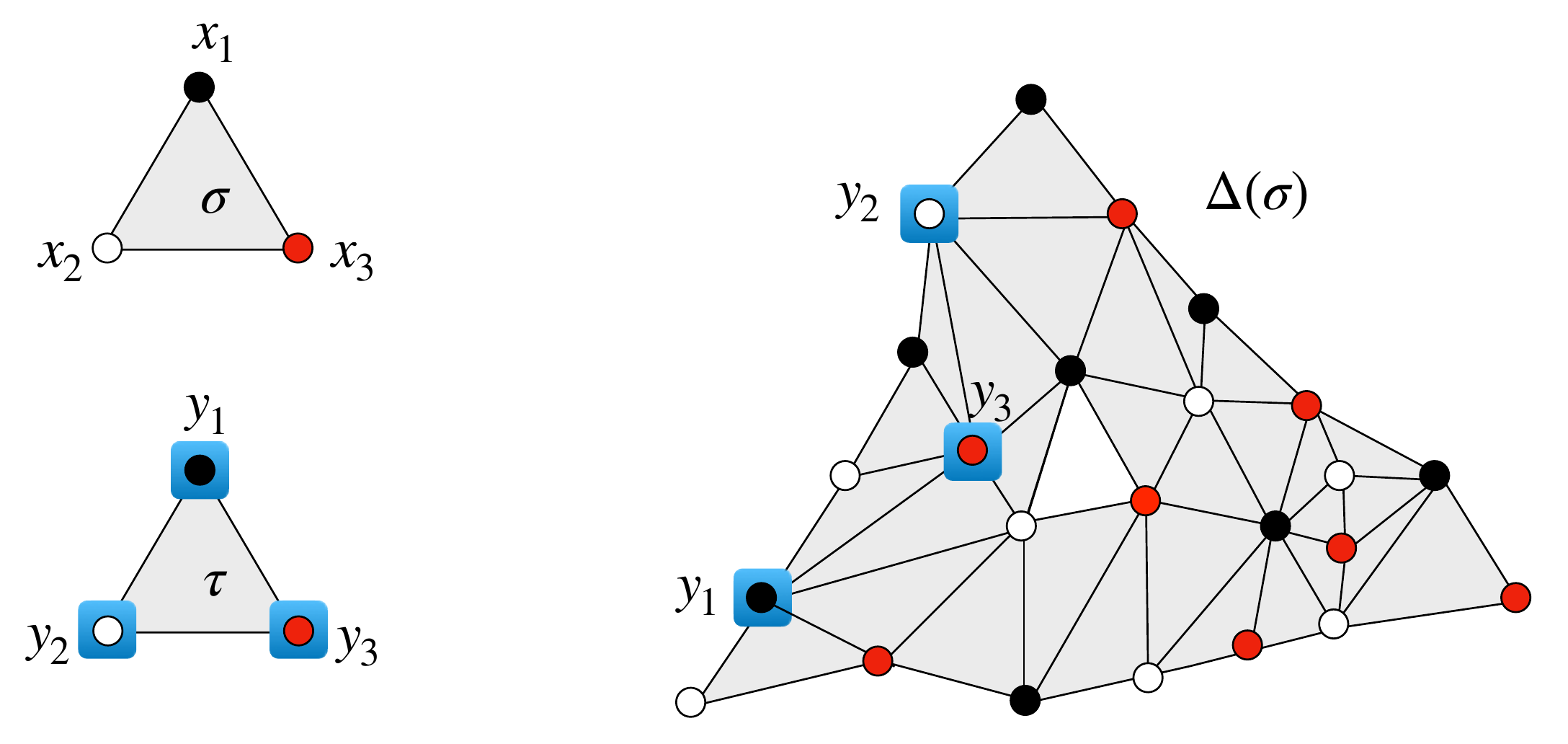}
	\caption{\sl Local task}
	\label{fig:TaskDelta1}
\end{figure}

\begin{definition}\label{def:localtask}
	Given a task $\Pi=(\m{I},\m{O},\Delta)$, a simplex $\sigma\in\m{I}$, and a chromatic set $\tau\subseteq V(\Delta(\sigma))$ satisfying 
	$\ID(\tau)= \ID(\sigma)$, the \emph{local task for $\sigma$ and $\tau$} is the task $\Pi_{\tau,\sigma}=(\tau,\Delta(\sigma),\D_{\tau,\sigma})$ where $\D_{\tau,\sigma}:\tau\to 2^{\Delta(\sigma)}$ is the  map 
	defined by:
	\begin{enumerate}
		\item for every vertex $v\in \tau$:  $\D_{\tau,\sigma}(v)=\{v\}$; 
		\item for every $\tau'\subseteq \tau$ with $|\tau'|>1$: 
		$
		\D_{\tau,\sigma}(\tau') = \proj_{\id(\tau')} \big ( \Delta(\sigma)\big).
		$
	\end{enumerate}
\end{definition}

Figure~\ref{fig:TaskDelta1} provides an illustration of a local task, and in particular of the map~$\D_{\tau,\sigma}$. 
This map seems similar to~$\Delta$ restricted to~$\sigma$, but there are several differences that should be noted.
First, while~$\Delta$ takes as inputs simplices of the \emph{input complex}~$\m{I}$, the map~$\D_{\tau,\sigma}$ takes as inputs subsets of~$\tau$, which are sets of vertices in the \emph{output complex}~$\m{O}$;
both maps return sets of simplices of the output complex $\m{O}$.
Second, $\D_{\tau,\sigma}$ is more restrictive than~$\Delta$ for solo executions, as the vertices of~$\tau$ are fixed by~$\D_{\tau,\sigma}$ --- a process~$i$ of~$\tau$ running solo must output its value~$y_i$ in~$\tau$. 
And third, $\D_{\tau,\sigma}$ is more liberal than $\Delta$ in the other executions, as for a face~$\tau'$ of~$\tau$ with dimension greater than~$0$, the image of~$\tau'$ by~$\D_{\tau,\sigma}$ is any simplex of~$\Delta(\sigma)$ colored with the IDs of the processes in~$\tau'$. 


\paragraph{Remark.} 

Given the local task $\Pi_{\tau,\sigma}=(\tau,\Delta(\sigma),\D_{\tau,\sigma})$ and $\tau'\subseteq \tau$, the set $\D_{\tau,\sigma}(\tau')$ induces a subcomplex of $\Delta(\sigma)$, potentially reduced to a single vertex whenever $|\tau'|=1$. In particular, if $\tau$ is not a single vertex then $\D_{\tau,\sigma}(\tau)=\Delta(\sigma)$. 

\medskip 

We are now ready to define the closure task $\cl_M(\Pi)$,
which will be a task that is expected to be easier than~$\Pi$ for model~$M$. 
Given an input simplex~$\sigma$ for~$\Pi$, the closure of~$\Pi$ allows to output chromatic sets of vertices $\tau\subseteq V(\Delta(\sigma))$ which are not simplices of~$\Delta(\sigma)$. 
However, this set~$\tau$ must be ``close to each other'' in~$\Delta(\sigma)$, in the sense that the local task $\Pi_{\tau,\sigma}$ can be solved in one round w.r.t. a given model $M$. 


\begin{definition}\label{def:closure}
	The \emph{closure} of a task $\Pi=(\m{I},\m{O},\Delta)$  with respect to a communication model~$M$ is the task $\cl_M(\Pi)=(\m{I},\m{O}',\Delta')$ where $V(\m{O}')=V(\m{O})$, and, for any $\sigma\in \m{I}$ and any set $\tau\subseteq V(\m{O})$,
	we let 
	$\tau\in \Delta'(\sigma)$	
	if 
	$\tau$ is chromatic, 
	$\ID(\tau)= \ID(\sigma)$,  
	$\tau\subseteq V(\Delta(\sigma))$, and the local task $\Pi_{\tau,\sigma}$ is solvable in at most one round in $M$.
	The simplices of $\m{O}'$ are the images of $\Delta'$, and all their faces.
\end{definition}

Figure~\ref{fig:closure} illustrates Definition~\ref{def:closure},
and in particular of how the map~$\Delta'$ is built. In this figure, the closure is with respect to the immediate snapshot model, as witnessed by the fact that the complex~$\m{P}^{(1)}(\tau)$ is the standard chromatic subdivision.
We have $\tau\in \Delta'(\sigma)$ because the local task $\Pi_{\tau,\sigma}=(\tau,\Delta(\sigma),\D_{\tau,\sigma})$ is solvable in one round with immediate snapshot, as there is a simplicial map $f:\m{P}^{(1)}(\tau)\to\Delta(\sigma)$ that respects~$\Delta_{\tau,\sigma}$: the chromatic subdivision of~$\tau$ is mapped to the dark subcomplex of $\Delta(\sigma)$ (the small letters specify~$f$). Note that the closure of a task depends on the communication model, as a local task $\Pi_{\tau,\sigma}$ may be solvable in one round of some model $M$, but not of some other model~$M'$.

\begin{figure}[tb]
	\centering
	\includegraphics[width=11cm]{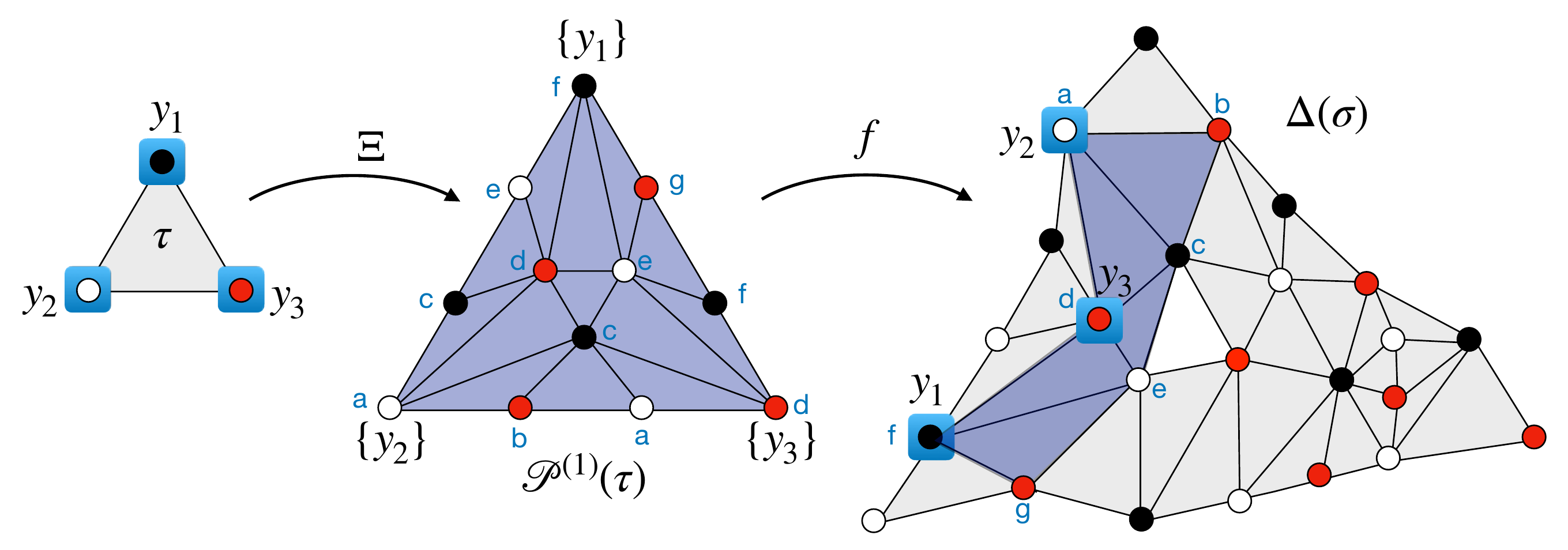}
	\caption{\sl  Closure of a task}
	\label{fig:closure}
\end{figure}

\paragraph{Remark.} 
Given a task $\Pi=(\m{I},\m{O},\Delta)$ and its closure $\cl_M(\Pi)=(\m{I},\m{O}',\Delta')$ w.r.t.~some model~$M$,
we have ${\Delta(\sigma)\subseteq\Delta'(\sigma)}$ for every $\sigma\in \m{I}$.
To see this, 
consider some $\sigma\in \m{I}$ and let $\tau\subseteq V(\m{O})$ be a chromatic set satisfying $\ID(\tau)= \ID(\sigma)$. 
If $\tau\in\Delta(\sigma)$ (that is, $\tau$ is a simplex of $\Delta(\sigma)$) then the local task $\Pi_{\tau,\sigma}=(\tau,\Delta(\sigma),\D_{\tau,\sigma})$ is solvable in zero rounds, by having each process output its input.
\medbreak

Since ${\Delta(\sigma)\subseteq\Delta'(\sigma)}$ for every  $\sigma\in \m{I}$,  the closure $\cl_M(\Pi)=(\m{I},\m{O}',\Delta')$ of a task $\Pi=(\m{I},\m{O},\Delta)$  is  not more difficult that the task itself.  In the next subsection, we will show that the closure of a task is in fact \emph{simpler} than the original task. 

\subsection{Asynchronous Speedup Theorem}
We say that an iterated model $M$ \emph{allows solo executions},  if in every round $t$ of Algorithm~\ref{alg:generic-model}, for each process~$i$,
there is an execution where the operations of process $i$ in round $t$ take place before the
operations of all the other processes in this round. Thus, 
there is an execution where process $i$ reads only its own write in its collect operation.

\begin{theorem}\label{thm:speedup}
	Let $M$ be an iterated model allowing solo executions.  For every $t\geq 1$, if  a task $\Pi$ is solvable in $t$ rounds in~$M$, then the closure of~$\Pi$ with respect to~$M$ is solvable in $t-1$ rounds in~$M$. 
\end{theorem}

\begin{proof}
	Let $\Pi=(\m{I},\m{O},\Delta)$ be a task on $n$ processes, let $t\geq 1$, and assume that  $\Pi$ is solvable in $t$ rounds in model~$M$. Let $\cl_M(\Pi)=(\m{I},\m{O}',\Delta')$ be the closure of~$\Pi$ w.r.t.~$M$. Our goal is to show that $\cl_M(\Pi)$ is solvable in $t-1$ rounds in~$M$. 
	
	Since  $\Pi$ is solvable in $t$ rounds, there exists a simplicial map 
	$
	f: \m{P}^{(t)} \to \m{O}
	$
	that agrees with~$\Delta$. We aim at defining a simplicial map 
	$
	f': \m{P}^{(t-1)} \to \m{O}',
	$
	that agrees with~$\Delta'$. To this end, let 
	$
	(i,V_i)\in \m{P}^{(t-1)}
	$
	be a vertex of $\m{P}^{(t-1)}$, i.e., $V_i$ is a possible view of process~$i\in [n]$ after $t-1$ rounds (see Figure~\ref{fig:MapDelta1}). We merely define 
	$
	f'(i,V_i)=f(i,\{(i,V_i)\}). 
	$
	That is, $f'(i,V_i)$ is equal to the image by $f$ of the vertex $(i,\{(i,V_i)\})\in \m{P}^{(t)}$ resulting from $(i,V_i)\in \m{P}^{(t-1)}$ whenever process~$i$ runs solo at round~$t$. Note that $f(i,\{(i,V_i)\})$ is a vertex of $\m{O}$, and thus $f'(i,V_i)$ is a vertex of $\m{O}'$, as $V(\m{O}')=V(\m{O})$.

	\begin{figure}[tb]
		\centering
		\includegraphics[width=10cm]{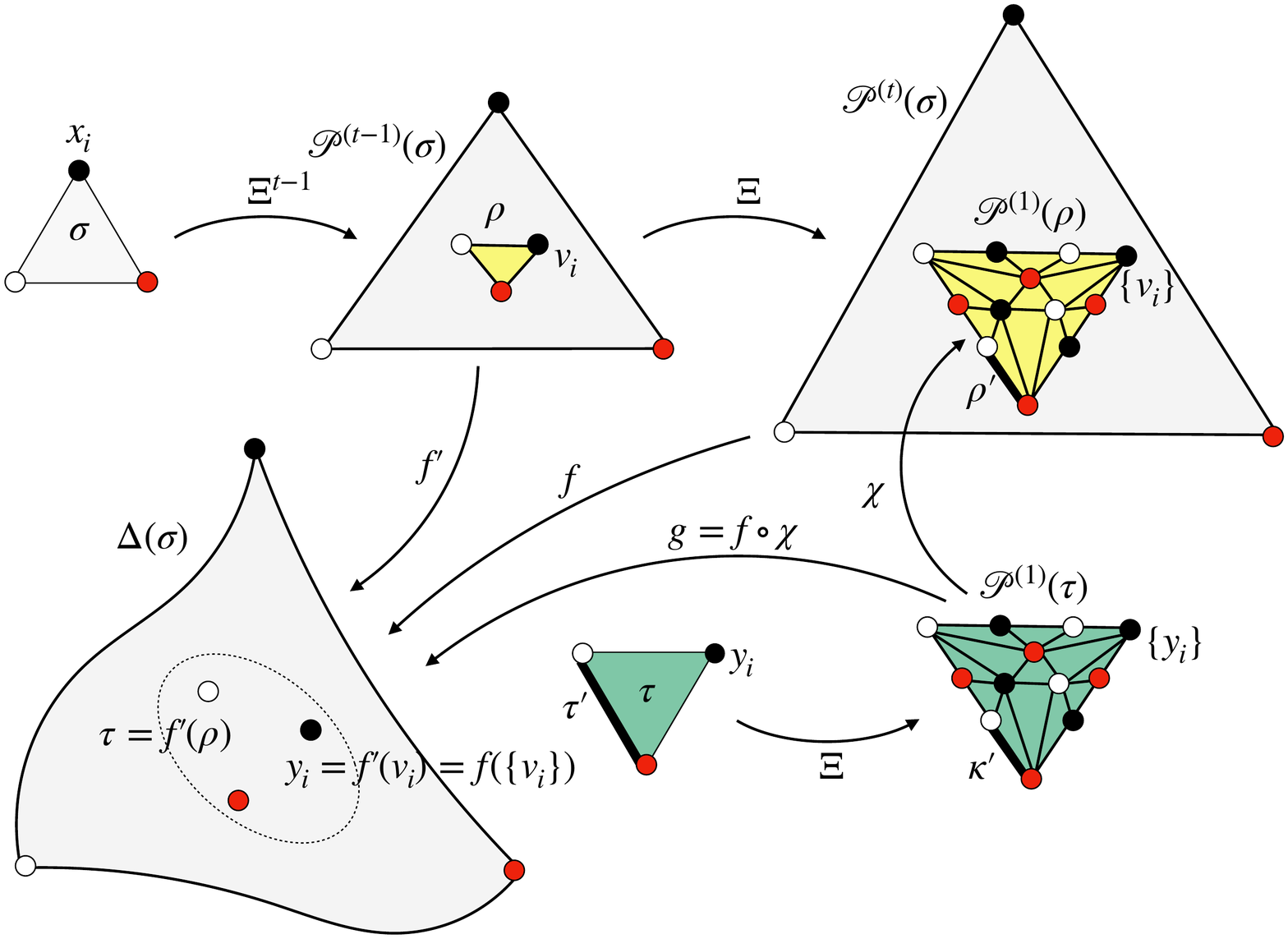}
		\caption{\sl Construction in the proof of Theorem~\ref{thm:speedup}}
		\label{fig:MapDelta1}
	\end{figure}

	Let us show that $f'$ is simplicial and agrees with~$\Delta'$. To this end, let $\sigma\in\m{I}$, and let 
	$$\rho\in  \m{P}^{(t-1)}(\sigma),$$ 
	i.e., $\rho$~results from input~$\sigma$ after $t-1$ rounds of communication (see Figure~\ref{fig:MapDelta1}). We aim at establishing that $f'(\rho)$ is a simplex of $\Delta'(\sigma)$.  Let us assume that 
	$
	\sigma=\{(i,x_i):i\in I\}
	$
	and
	$$
	\rho=\{(i,V_i):i\in I\},
	$$
	for some set $I\subseteq[n]$ of indices.
	Let us then consider $\tau=f'(\rho)$. For every $i\in I$, let $(i,y_i)=f'(i,V_i)$ be the output vertex returned by~$f'$ for process~$i$ with view~$V_i$ in~$\m{P}^{(t-1)}(\sigma)$. We  have 
	$
	\tau=f'(\rho)=\{f'(i,V_i):i\in I\}=\{(i,y_i):i\in I\}.
	$
	To establish that $f'$ is simplicial and agrees with~$\Delta'$, it is sufficient to show that $\tau$ is a simplex of~$\Delta'(\sigma)$. 
	
	To show that $\tau\in \Delta'(\sigma)$, it is sufficient to show that the local task $\Pi_{\tau,\sigma}=(\tau,\Delta(\sigma),\D_{\tau,\sigma})$ is solvable in one round (cf.~Definition~\ref{def:closure}),
	which we do next by presenting an algorithm.
	The task $\Pi_{\tau,\sigma}$ is specific for~$\tau$ and~$\sigma$, and that~$\tau$ was chosen by first choosing $\rho\in\m{P}^{(t-1)}(\sigma)$,
	so 
	when designing our algorithm we can use $\tau, \sigma$, and $\rho$.
	Moreover, we use the existence of a $t$-round algorithm for~$\Pi$, and the corresponding decision map~$f$.
	Intuitively, one can think of $\tau, \sigma, \rho$ and $f$ as hard-coded in the algorithm.

	Operationally, the following algorithm is solving  the local task $\Pi_{\tau,\sigma}$ in one round. 
	For each $i\in\ID(\tau)$, process~$i$ starts with input~$y_i$.
	It performs one communication round and collects a set $z_i=\{(j,y_j):j\in J_i\}$ with $i\in J_i\subseteq I$.
	To compute its output value, process~$i$ uses solely the set~$J_i$ of indices appearing in~$z_i$ (and not the values collected), and, using $\rho=\{(i,V_i):i\in I\}$, it computes the view $W_i=\{(j,V_j):j\in J_i\}$. 
	Then, process~$i$ outputs~$f(i,W_i)$. 
	
	Topologically, we define a decision map 
	$
	{g:\m{P}^{(1)}(\tau)\to \Delta(\sigma)}
	$
	as follows (see Figure~\ref{fig:MapDelta1}). 
	The vertices of $\m{P}^{(1)}(\tau)$ are in one-to-one correspondence with the vertices of $\m{P}^{(1)}(\rho)$, due to the canonical isomorphism~$\chi$ between $\m{P}^{(1)}(\tau)$ and~$\m{P}^{(1)}(\rho)$ (cf. Eq.~\eqref{eq:canonicaliso}), and so we define $g=f\circ\chi$. The map~$g$ is simplicial since $f$ and $\chi$ are simplicial. It remains to show that $g$ agrees  with~$\D_{\tau,\sigma}$,
	for which the two conditions of Definition~\ref{def:localtask} must be fulfilled. 
	First, let $(i,y_i)$ be a vertex of~$\tau$ for some $i\in \ID(\tau)$. We have 
	$$
	g(i,\{(i,y_i)\})=f(i,\{(i,V_i)\})=f'(i,V_i)=(i,y_i). 
	$$
	It follows that $g(i,\{(i,y_i)\})\in\D_{\tau,\sigma}(i,y_i)$ since, by definition, 
	$
	\D_{\tau,\sigma}(i,y_i)=\{(i,y_i)\}
	$, 
	and Condition~$1$ of Definition~\ref{def:localtask}
	is fulfilled. 
	To check Condition~$2$, let $\tau'\in\tau $ be a simplex of the input complex~$\tau$ of the local task $\Pi_{\tau,\sigma}$.
	In order to show that 
	$
	g(\m{P}^{(1)}(\tau'))\subseteq\D_{\tau,\sigma}(\tau'), 
	$
	consider a simplex $\kappa'\in \m{P}^{(1)}(\tau')$ with $\ID(\kappa')=\ID(\tau')$, and let $\rho'=\chi(\kappa')$ be the corresponding simplex in $\m{P}^{(1)}(\rho)$. 
	By the definition of the map $g$, we have $g(\kappa')=f\circ\,\chi(\kappa')=f(\rho')$. 
	Since $f$ solves the task~$\Pi$ with input~$\sigma\in \m{I}$, we have 
	$
	f(\rho')\in  \Delta(\sigma)
	$.
	Since $\ID(f(\rho'))=\ID(\rho')$, we get that $f(\rho')\in \proj_{\id(\tau')}(\Delta(\sigma))$ as desired, 
	and Condition~$2$ of Definition~\ref{def:localtask} holds as well.
	It follows that the map~$g$ does solve the local task $\Pi_{\tau,\sigma}$ in a single round, and therefore $\tau\in\Delta'(\sigma)$. As a consequence, $f'$ is simplicial and agrees with~$\Delta'$, which completes the proof. 
\end{proof}

\subsection{Impossibility of Consensus in the wait-free  IIS Model}
\label{app:application-to-consensus}

The speedup theorem enables to provide short proofs for classical impossibility results or lower bounds, as illustrated below for Consensus. A task $\Pi=(\m{I},\m{O},\Delta)$ is said to be a \emph{fixed-point} for model~$M$ whenever $\cl_M(\Pi)=\Pi$. 

\begin{lemma}\label{lem:fixed-point}
	A task $\Pi$ that is a fixed-point for model~$M$ is either solvable in zero rounds in~$M$, or unsolvable in~$M$. 
\end{lemma}

\begin{proof}
	Let us assume that there exists $t\geq 1$ such that $\Pi$ is solvable in $t$~rounds in~$M$. By the speedup theorem,  $\cl_M(\Pi)=\Pi$ is solvable in $t-1$ rounds in~$M$. By repeating the application of the speedup theorem to~$t-1,t-2,\dots,1$, we eventually get that $\Pi$ is solvable in zero rounds in~$M$. 
\end{proof}

Recall that binary consensus task for $n$ processes is the task $(\m{I},\m{O},\Delta)$ defined a follows. The input complex $\m{I}$ is composed of all simplices of the form $\sigma=\{(i,x_i):i\in I\}$ for some nonempty set $I\subseteq [n]$ of processes, where $x_i\in\{0,1\}$ for every $i\in I$.
The output complex $\m{O}$ has two facets $\tau_0$ and $\tau_1$, where, for $y\in \{0,1\}$, $\tau_y=\{(i,y):i\in[n]\}$. The set of legal output simplices for the input simplex $\sigma=\{(i,x_i):i\in I\}\in \m{I}$ is 
\[
\Delta(\sigma)=\left \{\begin{array}{ll}
	\{\proj_I(\tau_0),\proj_I(\tau_1)\} & \mbox{if there exists $i,j\in I$ such that $x_i\neq x_j$}\\
	\{\sigma\} & \mbox{otherwise.}
\end{array}\right.
\]

\begin{corollary}\label{cor:impossibility-consensus}
	Binary consensus is impossible to solve in the wait-free IIS model. 
\end{corollary}

\begin{proof}
	Let $\Pi$ denote the binary consensus task, and let $M$ denote the wait-free IIS model. For establishing the impossibility result, it is sufficient to show that ${\cl_M(\Pi)=\Pi}$. 
	After proving this, the impossibility of consensus is immediate from Lemma~\ref{lem:fixed-point}, and the fact that binary consensus is not solvable in zero rounds, simply because, in consensus, every process running solo must output its input value, and, in a zero-round algorithm, every process must output the same as if it was running solo.
	Hence, our goal is to show that $\Delta'(\sigma)=\Delta(\sigma)$ for every $\sigma=\{(i,x_i):i\in I\} \in \m{I}$. 
	Since, as noticed before, $\Delta(\sigma)\subseteq\Delta'(\sigma)$ for every task,  it suffices to show $\Delta'(\sigma)\subseteq\Delta(\sigma)$ for binary consensus.
	
	Let $\tau=\{(i,y_i):i\in I\} \subseteq V(\Delta(\sigma))$ be such that $\tau\in \Delta'(\sigma)$, and let us show that $\tau\in \Delta(\sigma)$. 
	Assume first that $\sigma=\{(i,x):i\in I\}$ for some $x\in\{0,1\}$. In this case, $\Delta(\sigma)=\sigma$, and thus, since $\tau\subseteq V(\Delta(\sigma))$, we have $\tau=\sigma$, and thus $\Delta'(\sigma)=\Delta(\sigma)$. 
	Assume now that $\sigma$ contains two vertices with distinct input values~$0$ and~$1$. 
	Our goal is to show that $\tau$ cannot include two vertices $(i,y_i)$ and $(j,y_j)$ with $y_i\neq y_j$. To this end, let $i,j\in I$ be two different indices. Since $\tau\in\Delta'(\sigma)$, the local task $\Pi_{\tau,\sigma}=(\tau,\Delta(\sigma),\D_{\tau,\sigma})$ is solvable in a single round. So, let $f:\m{P}^{(1)}(\tau)\to \Delta(\sigma)$ be a simplicial map that agrees with~$\D_{\tau,\sigma}$. 
	Consider the path
	\[
	(i,\{(i,y_i)\}) \; \rule[2pt]{20pt}{1pt} \; (j,\{(i,y_i),(j,y_j)\}) \; \rule[2pt]{20pt}{1pt} \; (i,\{(i,y_i),(j,y_j)\}) \; \rule[2pt]{20pt}{1pt} \;  (j,\{(j,y_j)\})
	\]
	connecting the vertices $(i,\{(i,y_i)\})$ and $(j,\{(j,y_j)\})$ in $\m{P}^{(1)}(\tau)$. This path is composed of three distinct edges, i.e., three distinct 1-dimensional simplices, $e_1,e_2$, and $e_3$. Since $f$ is simplicial, the image of this path by $f$ is a path connecting the vertices $f(i,\{(i,y_i)\})$ and $f(j,\{(j,y_j)\})$.  Since $f$ agrees with~$\D_{\tau,\sigma}$, it holds that $f(i,\{(i,y_i)\})=(i,y_i)$ and $f(j,\{(j,y_j)\})=(j,y_j)$. Moreover, each of the (not necessarily distinct) edges $f(e_k)$, $k\in\{1,2,3\}$, must belong to $\D_{\tau,\sigma}(\{(i,y_i),(j,y_j)\})=\proj_{\{i,j\}}(\Delta(\sigma))$. We have $\proj_{\{i,j\}}(\Delta(\sigma))=\{\proj_{\{i,j\}}(\tau_0),\proj_{\{i,j\}}(\tau_1)\}$.  It follows that either all edges $f(e_1)$, $f(e_2)$, and $f(e_3)$ are equal to $\proj_{\{i,j\}}(\tau_0)$, or all of them are equal to  $\proj_{\{i,j\}}(\tau_1)$. As a consequence, $y_i=y_j$. Therefore, all the output values in~$\tau$ are identical, and thus $\tau\in\Delta(\sigma)$, as desired. 
\end{proof}

\section{Extensions of the Speedup Theorem}
\label{sec:extSpeedup}

In order to design lower bounds or impossibility results to cases where processes have access to an object solving some specific tasks (e.g., test-and-set or binary consensus), we need to extend the speedup theorem to models that use not only registers but also other communication objects, which we informally call \emph{black box} objects. 

\subsection{Augmented Models}
%
\begin{wrapfigure}[12]{R}{0.35\textwidth}
	\center
	\resizebox{1\totalheight}{!}{
		\begin{minipage}{0.38\textwidth}
			\vspace*{-4ex}
			\small
			\begin{algorithm}[H]
				\SetAlgoLined
				\SetKwFor{loop}{\!}{for $r\gets 1$ to $t$}{end}
				$V_i \gets x_i$\\
				\loop{}{
					$\mathbf{write}\;(i,V_i)$ to $M_r[i]$\\
					$a_i\gets \alpha(i,V_i,r)$\\
					$b_i\gets \mathbf{B}_r(a_i)$\\
					$C_i\gets \mathbf{collect} \; M_r[1..n]$\\
					$V_i\gets (b_i,C_i)$
				}
				$y_i\gets f(i,V_i)$\\
				\textbf{output} $y_i$.
				\caption{\sl Code for process $i\in [n]$ with input~$x_i$}
				\label{alg:generic-extended}
			\end{algorithm}
		\end{minipage}
	}
\end{wrapfigure}
We consider generic round-based algorithms in a model augmented with a black box object $\mathbf{B}$ as displayed in Algorithm~\ref{alg:generic-extended}. 
This algorithm is similar to Algorithm~\ref{alg:generic-model} except that a call to a black box~$\mathbf{B}$ is inserted at every round between the write and collect instructions. 
The calls performed at two different rounds are independent, i.e., there are $t$ copies $\mathbf{B}_1,\dots,\mathbf{B}_t$ of $\mathbf{B}$, where $\mathbf{B}_r$ is the black box used by all processes at round~$r$. 
Some of the black box  objects we use have input and output values,
hence before invoking~$\mathbf{B}$ at a given round~$r$, each process~$i\in[n]$ computes an input value $a_i$ for $\mathbf{B}_r$, which is chosen by a function $\alpha$ that takes into account the process $i$, its view $V_i$, and the round number $r$. 
When process~$i$ invokes $\mathbf{B}_r$ with input~$a_i$, it gets a value~$b_i$ in return. To complete the round, process~$i$ collects the view of the other processes, and forms a new view which is the pair formed by the value obtained from $\mathbf{B}_r$, and the collection of the views of the other processes.

\paragraph{Remark.} In this paper, we assume that the black box object $\mathbf{B}$ is \emph{consistent} in the sense that, for the same inputs to the processes, and for the same interleaving of the reads and writes by the processes, $\mathbf{B}$~returns the same outputs to the processes. For instance, for the binary consensus box where some processes provide the box with input~0, and some others provide the box with input~1, the box will systematically produce the same outputs, either all~0s, or all~1s, for the same interleaving. We may assume consistency since we are only interested in lower bounds.  

\subsection{Extended Speedup Theorem}
We say that an iterated model $M$ augmented with a black box \emph{allows solo executions},  if in every round $t$ of Algorithm~\ref{alg:generic-extended}, for each process $i$,
there is an execution where the operations of process $i$ in round $t$ take place before the
operations of all other processes in this round.

\begin{theorem}\label{thm:speedup-extended}
	Let $M$ be an iterated model augmented with a black box allowing solo executions.  For every $t\geq 1$, if  a task $\Pi$ is solvable in $t$ rounds in~$M$, then the closure of~$\Pi$ with respect to~$M$ is solvable in $t-1$ rounds in~$M$.
\end{theorem}

\begin{proof}
	The proof is almost identical to the proof of Theorem~\ref{thm:speedup}, so we just underline the changes. Given a simplicial map 
	$
	f: \m{P}^{(t)} \to \m{O}
	$
	that agrees with~$\Delta$, we define a simplicial map 
	$
	f': \m{P}^{(t-1)} \to \m{O}',
	$
	that agrees with~$\Delta'$. Let 
	$
	(i,V_i)\in \m{P}^{(t-1)}
	$
	be a vertex of $\m{P}^{(t-1)}$. We define
	$
	f'(i,V_i)=f\big(i,(b_i,\{(i,V_i)\})\big), 
	$
	where $b_i=\mathbf{B}_t(a_i)$ is the output of the black box~$\mathbf{B}$ invoked at round~$t$ with $a_i=\alpha(i,V_i,t)$
	as input, whenever process~$i$ is running solo in~$\mathbf{B}$.
	Let us show that $f'$ is simplicial and agrees with~$\Delta'$. Let $\sigma=\{(i,x_i):i\in I\}\in\m{I}$ and  $\rho=\{(i,V_i):i\in I\}\in  \m{P}^{(t-1)}(\sigma)$, and let $\tau=f'(\rho)=\{(i,y_i):i\in I\}$. To show that $\tau\in \Delta'(\sigma)$, let us consider the local task $\Pi_{\tau,\sigma}=(\tau,\Delta(\sigma),\D_{\tau,\sigma})$, and let us check that it is indeed solvable in one round. As in the proof of Theorem~\ref{thm:speedup}, we define 
	$
	g:\m{P}^{(1)}(\tau)\to \Delta(\sigma)
	$
	as $g=f\circ\chi$ where $\chi:\m{P}^{(1)}(\tau)\to \m{P}^{(1)}(\rho)$ is the canonical isomorphism.  More specifically, every process $i\in I$ writes $y_i$ in memory, calls the black box with input $a_i=\alpha(i,V_i,t)$ to get a value~${b_i=\mathbf{B}_1(a_i)=\mathbf{B}_t(a_i)}$ and collects a set $z_i=\{(j,y_j):j\in J_i\}$ for some set~$J_i$ with $i\in J_i\subseteq I$. Process~$i$ outputs 
	$
	g\big(i,(b_i,z_i)\big)=f\big(i,(b_i,W_i)\big)
	$
	where $W_i=\{(j,V_j):j\in J_i\}$. 
	The fact that $g$ is simplicial and agrees with~$\D_{\tau,\sigma}$ follows exactly from the same arguments as in the proof of Theorem~\ref{thm:speedup}, using the fact that the box is consistent
	--- which guarantees that, for every $\kappa\in\m{P}^{(1)}(\tau)$,  the outputs of the box for $\kappa$ and for $\chi(\kappa)\in \m{P}^{(1)}(\rho)$ are identical. 
\end{proof}

\subsection{Impossibility of Consensus in the Wait-Free IIS Model  with \TS}
\label{app:impossibility-of-consensus-with-TS}

Recall that  
\TS{}  is an object that requires no inputs, and every process invoking \TS\/ gets a value~0 or~1, with the guarantee that only the first process to invoke it gets a~1, all the other getting~0.
It is known that \TS\/ has consensus number~2~\cite{Herlihy91}, i.e., it can be used to solve binary consensus among two processes, but not among more processes. 
Multi-value (i.e., not only binary) consensus among two processes can actually be solved in a single round with \TS:
A process receiving the value~1 from \TS\/ outputs its own input value, and a process receiving the value~0 from \TS\/ outputs the input value of the other process. 
Note that a process~$i$ receiving the value~0 from \TS\/ has the guarantee that the input value of the other process was written in memory before process~$i$ performs a snapshot, as if it was not running \TS\/ solo (if the process would have run \TS\/ solo, it would have obtained the value~1 from \TS). 

Figure~\ref{fig:2-proc-consensus-solvable} illustrates why 2-process binary consensus is solvable with access to a \TS\/ object. After one round, every process $i\in\{1,2\}$ which sees only itself when reading the shared memory must win in \TS, and thus gets a view $(1,\{(i,x_i)\})$, where $x_i\in\{0,1\}$ is the input of process~$i$. Every process~$i$ which saw the other process~$j$ gets a view $(b_i,\{(i,x_i),(j,x_j)\})$ where $x_i$ and $x_j$ are the input values of processes~$i$ and~$j$, respectively, and $b_i\in\{0,1\}$ is the value produced by \TS\/ for process~$i$. The mapping $f:\m{P}^{(1)}\to\m{O}$ displayed on Figure~\ref{fig:2-proc-consensus-solvable} is simplicial, and its composition with the map 
$\Xi:\m{I}\to \m{P}^{(1)}$
agrees with the specification of consensus.

\begin{figure}[tb]
	\centering
	\includegraphics[width=12cm]{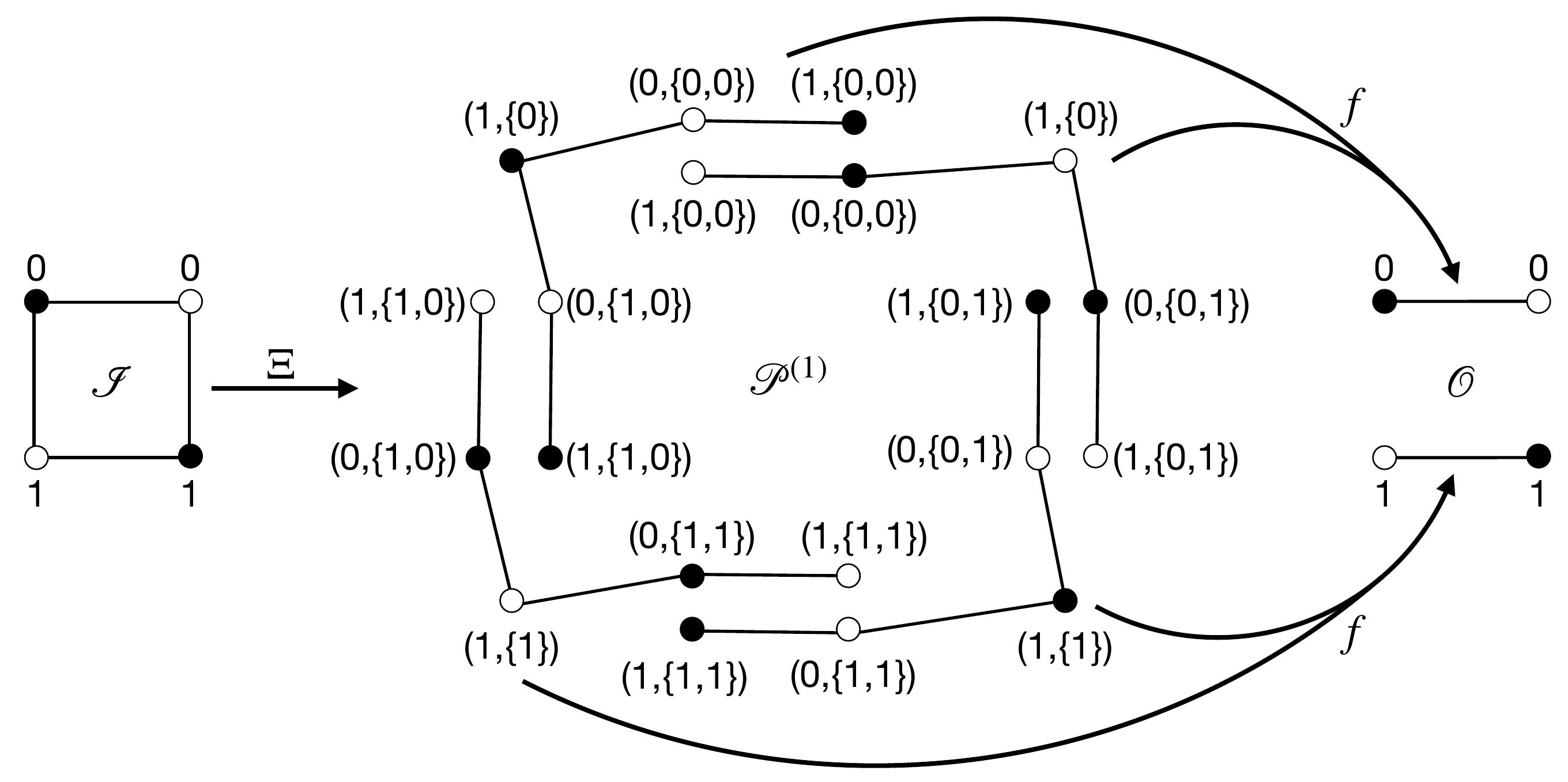}
	\caption{\sl 2-process binary consensus is solvable using \TS}
	\label{fig:2-proc-consensus-solvable}
\end{figure}

In contrast, we show that for more than two processes, binary consensus is not solvable even using \TS\/, by applying Theorem~\ref{thm:speedup-extended}. Figure~\ref{fig:prot-compl-test-and-set} describes the protocol complex $\m{P}^{(1)}(\sigma)$ for a 2-dimensional simplex~$\sigma$ for the immediate snapshot model augmented with \TS. For each $i\in\{1,2,3\}$, instead of four vertices with same ID~$i$ as in the 12-vertex chromatic subdivision of the triangle resulting from immediate snapshot, Figure~\ref{fig:prot-compl-test-and-set} displays seven vertices with the same ID~$i$, each vertex of the chromatic subdivision being duplicated depending on the outcome~0 or~1 of \TS, except the vertex corresponding to a solo execution for which the outcome of \TS\/ is necessarily~1. 

\begin{figure}[tb]
	\centering
	\includegraphics[width=10cm]{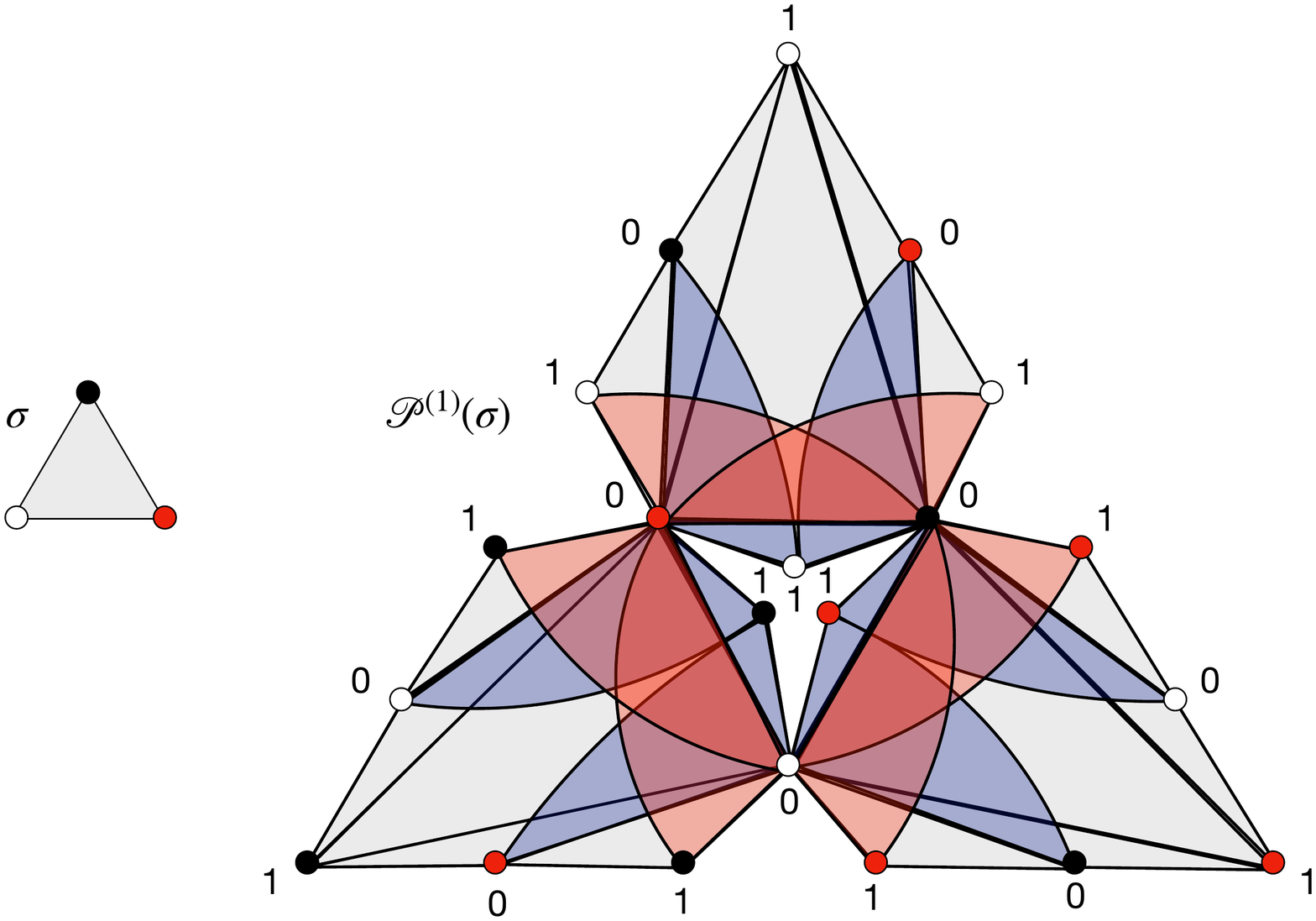}
	\caption{\sl 1-round protocol complex for three processes using \TS}
	\label{fig:prot-compl-test-and-set}
\end{figure}

\begin{corollary}
	\label{cor: consensus impossibility with ts}
	For every $n>2$, binary consensus among $n$ processes is impossible to solve  in the wait-free IIS model augmented with \TS.
\end{corollary}

\begin{proof}
	Let $M$ be the model of iterated immediate snapshot with a \TS{} object. 
	We cannot proceed exactly as in the proof of Corollary~\ref{cor:impossibility-consensus}, as the binary consensus task $\Pi_{\text{\tiny CON}}=(\m{I},\m{O},\Delta)$ for $n>2$ processes is not formally a fixed point for~$M$. 
	The reason for this is that, in $\Pi_{\text{\tiny CON}}$, even if only two processes~$i$ and~$j$ participate, they must agree on the same value, while in the closure task $\cl_M(\Pi_{\text{\tiny CON}})=(\m{I},\m{O}',\Delta')$ this is not the case. 
	Formally, since $2$-process consensus is solvable using \TS{} in one round, for $\sigma=\{(i,0),(j,1)\}$ we have  $\Delta'(\sigma)=\big\{\{(i,y_i),(j,y_j)\}:(y_i,y_j)\in\{0,1\}^2\big\}$. 
	
	Let us instead consider the simpler task $\Pi=(\m{I},\m{O},\Delta)$, where the input and output values and complexes are as in consensus, the validity condition must still be fulfilled (every output value must be an input value),
	but, only if at least three processes participate, they must output the same value.
	Every consensus algorithm also solves this problem, so it is enough to prove the impossibility of $\Pi$.
	We now prove that $\Pi$ is a fixed point of $M$.
	
	Let $\cl_M(\Pi)=(\m{I},\m{O}',\Delta')$ be the closure task of $\Pi$. Let us 
	fix simplices $\sigma=\{(i,x_i)\mid i\in I\}\in\m{I}$ and $\tau=\{(i,y_i),\mid i\in I\}\in\Delta'(\sigma)$
	for a non-empty set $I\subseteq[n]$.
	If $|I|\le 2$,  then
	any subset of $ \tau'\subseteq V(\Delta(\sigma))$ is already a simplex of $\Delta(\sigma)$, and we are done.
	Let us now consider the case of $|I|\geq 3$. 
	Let $i,j,k\in I$, and let $y_i,y_j,y_k$ be the values associated with the three processes~$i,j,k$ in~$\tau$.
	Our goal is to show that $y_i=y_j=y_k$, which will imply that $\tau\in \Delta(\sigma)$.
	
	By the definition of a closure of a task,
	we know that the local task $\Pi_{\tau,\sigma}=(\tau,\Delta(\sigma),\D_{\tau,\sigma})$ is solvable in one round.
	Let $f:\m{P}^{(1)}\to \m{O}$
	be a simplicial map solving $\Pi_{\tau,\sigma}$.
	For every $v=(\ell,y_\ell)\in\tau$, $\ell\in I$, we have $\D_{\tau,\sigma}(v)=\{v\}$, 
	hence $f\big(\ell,(t_\ell,\{\ell,y_\ell)\})\big)= (\ell,y_\ell)$, where~$t_\ell$ is the output of \TS, which is actually equal to~$1$ since the vertex $\big(\ell,(t_\ell,\{(\ell,y_\ell)\})\big)$ represents a solo step by process~$\ell$.
	
	In $\m{P}^{(1)}(\sigma)$, let us consider the two simplices (see Figure~\ref{fig:n-proc-consensus-not-solvable})
	\[
	\rho_{i,j,k}=\Big\{
	\big(i,(1,\{(i,y_i)\})\big),  \big(j,(0,\{(i,y_i),(j,y_j)\})\big),     
	\big(k,(0,\{(i,y_i),(j,y_j),(k,y_k)\})\big)
	\Big\},
	\]
	and 
	\[
	\rho_{j,i,k}=\Big\{
	\big(j,(1,\{(j,y_j)\})\big),  \big(i,(0,\{(i,y_i),(j,y_j)\})\big),     
	\big(k,(0,\{(i,y_i),(j,y_j),(k,y_k)\})\big)
	\Big\}.
	\]
	Let $\gamma=f\big(k,(0,\{(i,y_i),(j,y_j),(k,y_k)\})\big)$ be the output of $f$ on the vertex $\big(k,(0,\{(i,y_i),(j,y_j),(k,y_k)\})\big)$ of $\m{P}^{(1)}(\sigma)$.
	Since $f$ agrees with $\Delta_{\tau,\sigma}$, we have that $f(\rho_{i,j,k})\subseteq\proj_{\{i,j,k\}}(\Delta(\sigma))$, and the definition of $\Delta$ implies that, in $\Delta(\sigma)$, all processes output the same value $y_i=\gamma$.
	A similar argument on $f(\rho_{j,i,k})$ implies that $y_j=\gamma$, i.e., $y_i=y_j$.
	By applying this argument for every $\ell\in I$, we get that all processes decide the same value 
	$\gamma$ in $\tau$, that is, 
	$\tau=\{(\ell,\gamma)\mid \ell\in I\}$, which implies 
	that $\tau\in\Delta(\sigma)$, as claimed.
	
	The rest of the proof is similar to the proof of Lemma~\ref{lem:fixed-point}, but using 
	Theorem~\ref{thm:speedup-extended}. That is, 
	if $\Pi$ was solvable in $t$ rounds in $M$ among $n\geq 3$ processes, then $\Pi$ would have also been solvable in $t-1$ rounds, and, similarly, in $0$~rounds.
	Since this is not the case, $\Pi$ is not solvable in $M$ at all.
	As mentioned above, $\Pi$ is simpler than consensus, and so, for $n\geq 3$ processes,
	consensus is not solvable.
\end{proof}

\begin{figure}[tb]
	\centering
	\includegraphics[width=9cm]{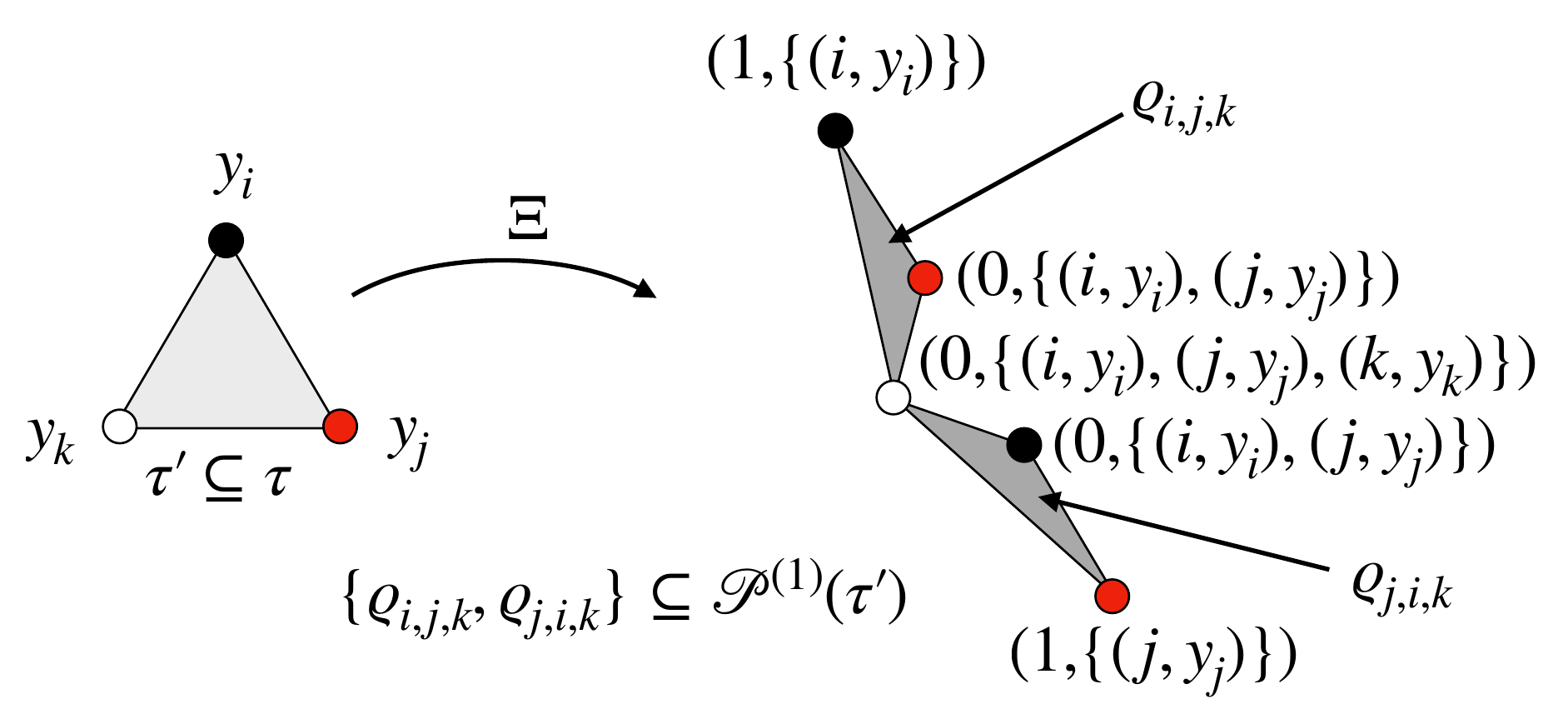}
	\caption{\sl Consensus is not solvable among $n>2$ processes, even using \TS}
	\label{fig:n-proc-consensus-not-solvable}
\end{figure}

\section{Lower Bounds for Approximate Agreement}
\label{sec:lowerApp}

In this section, we apply our speedup theorem, and its extension, to the \emph{approximate agreement} task.
This task is a relaxation of the consensus problems, where processes have to decide on values that in the interval defined by the input values, 
and are close to one another
(but not necessarily identical as in consensus).
Given a (small) real value $\epsilon>0$, in order to accurately define the $\epsilon$-approximate agreement task while using only finite topological objects, we assume that $\epsilon$ and all the input values are rational numbers in the interval $[0,1]$. Specifically, we discretize this interval as follows.
We choose a (large) integer $m\geq 1/\epsilon$ such that $\epsilon$ is an integral multiple of $1/m$ (i.e., $m\epsilon$ is an integer), and so are all the possible input values of the task.
We make sure that all the output values are also integral multiples of $1/m$. To this end, we avoid the usage of averaging algorithms, which are common in other works on approximate agreement.
We formally define $\epsilon$ approximate agreement as follows.

\begin{definition}\label{def:approx-agree}
	Given a parameter $\epsilon\in(0,1]$ which is an integral multiple of $1/m$ for some integer $m\geq 1$, 
	the \emph{$\epsilon$-approximate agreement} 
	for $n$ processes is the task $(\m{I},\m{O},\Delta)$, where
	\begin{align*}
		&\m{I}
		=\big \{\{(i,x_i):i\in I\}\mid 
		(\varnothing\neq I\subseteq [n]) \;\wedge\; (\forall i\in I, \; x_i\in\{0,1/m,2/m,\ldots,1\})\big\}\\
		&\m{O}
		=\big\{ \{(i,y_i):i\in I\}\mid
		(\varnothing\neq I\subseteq [n]),\;  (\forall i\in I, 
		y_i\in\{0,1/m,2/m,\ldots,1\})\; \wedge \; (\forall i,j\in I, |y_i-y_j|\leq\epsilon)\big\}\\
		&\Delta:\m{I}\to 2^{\m{O}}
		\text{ where } \Delta\big(\{(i,x_i):i\in I\}\big)= 
		\big\{ \{(i,y_i):i\in I\}\in\m{O}\mid \forall j\in I, \min\{x_i:i\in I\}\leq y_j\leq\max\{x_i:i\in I\} \big\}
	\end{align*}
\end{definition}

We are interested in solving approximate agreement when the processes have access to an object that can be used in order to solve consensus among two processes, and hence also $\epsilon$-approximate agreement for every $\epsilon$.
This creates a technical difficulty when applying our speed up technique, even when more than two processes are present. 
To overcome this, we define a relaxed version of $\epsilon$-approximate agreement,
where 
if exactly two processes participate, then they only need to output values in the range of inputs, but there is no bound on the distance between them.
\begin{definition}\label{def:liberal}
	The \emph{liberal version of $\epsilon$-approximate agreement} is
	identical to the $\epsilon$-approximate agreement task, except that 
	the output complex $\m{O}$ contains also all $1$-dimensional chromatic simplices, regardless of the distances between their values.
	The map $\Delta$ specifying the task is formally defined the same, which means that for $1$-dimensional input simplices it now contains more allowed output simplices.
\end{definition}
In fact, this is the type of relaxation we apply on the consensus task in the proof of Corollary~\ref{cor: consensus impossibility with ts}.
Note that any algorithm solving $\epsilon$-approximate agreement also solves the liberal version of $\epsilon$-approximate agreement, hence any lower bound for the latter implies a lower bound for the former.

\subsection{Approximate Agreement in the Wait-Free IIS Model}
\label{app:impossibility-approx-aggreement-IIS}

Using our tools, we can easily 
reprove the fact that $\epsilon$-approximate agreement takes $\Omega(\log 1/\epsilon)$ rounds in the standard IIS model~\cite{HoestS06}. For this, we show that the closure of $\epsilon$-approximate agreement  for two processes is $(3\epsilon)$-approximate agreement, and $(2\epsilon)$-approximate agreement for more processes. 
We start by proving some claims --- the first of them is almost a triviality.

\begin{claim}\label{claim:epprox-agree-zero-rounds}
	For $\epsilon<1$, 
	$\epsilon$-approximate agreement is not solvable in $0$ rounds, and so is the liberal version of $\epsilon$-approximate agreement for $n\geq 3$ processes.
\end{claim}

\begin{proof}
	Let $\epsilon<1$,
	and assume for contradiction that there is a $0$-round $\epsilon$-approximate agreement algorithm,
	or a $0$-round algorithm solving the liberal version of $\epsilon$-approximate agreement among $n\geq 3$ processes.
	Consider a setting with inputs $(1,0)$ and $(2,1)$, i.e., where process $1$ has input $0$ and process $2$ has input $1$.
	In a solo-execution each process must output its own input, and 
	the setting where several processes run is indistinguishable to them from the one where each runs solo, so they will output $0$ and $1$ in all cases, contradicting the task definition.
	Formally, we have $\Delta(\{(1,0)\})=\{(1,0)\}$ and $\Delta(\{(2,1)\}))=\{(2,1)\}$.
	As $\Xi:\m{I}\to\m{P}^{(0)}$ is simply the identity map, in a $0$-round protocol any decision map $f$ goes directly from $\m{I}$ to $\m{O}$.
	Hence, any decision map $f:\m{I}\to \m{O}$ must satisfy $f(1,0)=(1,0)$ and $f(2,1)=(2,1)$.
	
	For the original $\epsilon$-approximate agreement task,
	this implies that the input simplex $\sigma=\{(1,0),(2,1)\}$, is mapped to  $f(\sigma)=\{(1,0),(2,1)\}\notin\Delta(\sigma)$
	(in fact, we even have $f(\sigma)\notin\m{O}$).
	For the liberal version of the task and $n\geq 3$, consider the same setting but with an extra ``dummy'' process $3$, with input $0$.
	As before, $f(3,0)=(3,0)$,
	and  $f(\{(1,0),(2,1),(3,0)\})=\{(1,0),(2,1),(3,0)\}\notin\Delta(\sigma)$,
	where here $\Delta$ is the specification of the liberal version of $\epsilon$-approximate agreement.
	Hence, in both cases no algorithm and decision map $f$ can solve the problem.
\end{proof}

The next claim is about the closure of approximate agreement in the specific case of two processes.

\begin{claim}\label{claim:epprox-agree-n=2-one-round}
	If $\Pi$ is the $\epsilon$-approximate agreement task for two processes,
	then $\cl_M(\Pi)$ w.r.t. the wait-free IIS model~$M$ is the $(3\epsilon)$-approximate agreement task for two processes.
\end{claim}

\begin{proof}
	Let $\Pi_\epsilon=(\m{I},\m{O}_\epsilon,\Delta_\epsilon)$ be the $\epsilon$-approximate agreement task for $n=2$ processes, and let $\cl_M(\Pi_\epsilon)=(\m{I},\m{O}'_\epsilon,\Delta'_\epsilon)$ be its closure in the considered model~$M$. Let us denote the $(3\epsilon)$-approximate agreement task by $\Pi_{3\epsilon}=(\m{I},\m{O}_{3\epsilon},\Delta_{3\epsilon})$. We want to show that, for every $\sigma\in\m{I}$, $\Delta'_\epsilon(\sigma)=\Delta_{3\epsilon}(\sigma)$. We consider separately the cases where $\sigma$ is a vertex, and $\sigma$ is an edge (recall that we are in the setting involving $n=2$ processes).

	Let $\sigma \in \m{I}$ be a vertex.
	Let $\tau\in\Delta'_\epsilon(\sigma)$, which immediately implies $\tau\in\Delta_\epsilon(\sigma)$. 
	As $\sigma$ is a vertex, we have  $\Delta_\epsilon(\sigma)=\{\sigma\}$, from which we derive 
	$\tau=\sigma$. 
	We also have $\Delta_{3\epsilon}(\sigma)=\{\sigma\}$, and hence
	$\tau\in \Delta_{3\epsilon}(\sigma)$. 
	Conversely, let $\tau\in \Delta_{3\epsilon}(\sigma)$, and again $\Delta_{3\epsilon}(\sigma)=\{\sigma\}$ implies $\tau=\sigma$. The local task $(\sigma,\Delta_\epsilon(\sigma),\Delta_{\sigma,\sigma})$ is trivially solvable in one round---it is even solvable in zero rounds using the mapping $f(\sigma)=\sigma$---and hence $\tau\in\Delta'_\epsilon(\sigma)$. 
	We conclude that for vertices, $\Delta'_\epsilon$ and $\Delta_{3\epsilon}$ are identical.

	Let $\sigma=\{(1,x_1),(2,x_2)\}\in\m{I}$ be an edge. We consider an edge $\tau=\{(1,y_1),(2,y_2)\}\in\Delta_\epsilon'(\sigma)$, and show
	$\tau\in\Delta_{3\epsilon}(\sigma)$.
	To this end, we show 
	(1)~$\min\{x_1,x_2\}\leq y_i\leq\max\{x_1,x_2\}$ 
	for all $i\in \{1,2\}$,
	and 
	(2)~$|y_1-y_2|\leq 3\epsilon$.
	Since $\tau\in\Delta_\epsilon'(\sigma)$,
	the local task $\Pi_{\tau,\sigma}=(\tau,\Delta_\epsilon(\sigma),\D_{\tau,\sigma})$ is solvable in one round, and let $f:\m{P}^{(1)}(\tau)\to \Delta_\epsilon(\sigma)$
	be a simplicial map solving it.
	Since $\tau\subseteq V(\Delta_\epsilon(\sigma))$, we know that every vertex $v=(i,y_i)\in\tau$, $i\in\{1,2\}$, satisfies 
	$\min\{x_1,x_2\}\leq y_i\leq\max\{x_1,x_2\}$, 
	and thus property~(1) holds.
	Let us denote by $\gamma_1,\gamma_2$ the two values satisfying
	\[
	\begin{array}{rcl}
		f(1,\{(1,y_1)\}) & = 		&(1,y_1),\\
		f(2,\{(1,y_1),(2,y_2)\}) & = 	&(2,\gamma_2),\\
		f(1,\{(1,y_1),(2,y_2)\}) & = 	&(1,\gamma_1),\\
		f(2,\{(2,y_2)\}) & = 		&(2,y_2).
	\end{array}
	\]
	Since $f$ solves the local task $\Pi_{\tau,\sigma}=(\tau,\Delta_\epsilon(\sigma),\D_{\tau,\sigma})$, we have 
	$|y_2-y_1|
	\leq |y_2-\gamma_1|+|\gamma_1-\gamma_2|+|\gamma_2-y_1|
	\leq 3\epsilon$. Property~(2) thus holds as well, implying $\tau\in\Delta_{3\epsilon}(\sigma)$. 
	
	Conversely, let us consider $\tau=\{(1,y_1),(2,y_2)\}\in\Delta_{3\epsilon}(\sigma)$,
	and let us assume, w.l.o.g., that $y_1\leq y_2$.  
	Let $f:\m{P}^{(1)}(\tau)\to \Delta_\epsilon(\sigma)$ defined as follows, by setting $z= \min\{y_2,y_1+\epsilon\}$, and 
	\begin{equation}
		\begin{array}{rcl}
			f(1,\{(1,y_1)\}) & = 		&(1,y_1),\\
			f(2,\{(1,y_1),(2,y_2)\}) & = 	& (2,z),\\
			f(1,\{(1,y_1),(2,y_2)\}) & = 	&(1,\min\{y_2,z+\epsilon\}),\\
			f(2,\{(2,y_2)\}) & = 		&(2,y_2).
		\end{array}
		\label{eq:algo-for-2-proc}
	\end{equation}
	The map $f$ solves $\epsilon$-approximate agreement whenever $y_2-y_1\leq 3\epsilon$, and therefore it solves the local task $\Pi_{\tau,\sigma}=(\tau,\Delta_\epsilon(\sigma),\D_{\tau,\sigma})$ in one round. It follows that $\tau\in\Delta'_\epsilon(\sigma)$. 
	We conclude that, for edges as well, $\Delta'_\epsilon$ and $\Delta_{3\epsilon}$ are identical. 
\end{proof}

The next claim is about the closure of the liberal version of approximate agreement in the case of three or more processes.

\begin{claim}\label{claim:epprox-agree-n>2-one-round}
	If $\Pi=(\m{I},\m{O},\Delta)$ is the liberal version of $\epsilon$-approximate agreement with $n\ge 3$ processes,
	then the closure $\cl_M(\Pi)$ of $\Pi$ with respect to the wait-free IIS model~$M$ is the liberal version of $(2\epsilon)$-approximate agreement.
\end{claim}

\begin{proof}
	Let $\Pi_\epsilon=(\m{I},\m{O}_\epsilon,\Delta_\epsilon)$ be the liberal version of $\epsilon$-approximate agreement with $n\geq 3$ processes, and let $\cl_M(\Pi_\epsilon)=(\m{I},\m{O}'_\epsilon,\Delta'_\epsilon)$. 
	Let us denote the liberal version of $(2\epsilon)$-approximate agreement by $\Pi_{2\epsilon}=(\m{I},\m{O}_{2\epsilon},\Delta_{2\epsilon})$. 
	We want to show that, for every $\sigma\in\m{I}$, $\Delta'_\epsilon(\sigma)=\Delta_{2\epsilon}(\sigma)$. If $\sigma$ is a vertex, this equality holds exactly by the same reasoning as in the proof of 
	Claim~\ref{claim:epprox-agree-n=2-one-round}. 
	We hence focus on the case where $|\sigma|\geq 2$.
	
	Fix an input simplex $\sigma=\{(i,x_i)\mid i\in I\}\in\m{I}$, for some non-empty set $I\subseteq[n]$.
	If $|I|=1$, i.e., $I=\{i\}$, then $\sigma=\{(i,x_i)\}\in\m{I}$,
	and $\Delta'_\epsilon(\{(i,x_i)\})=
	\{(i,x_i)\}=
	\Delta_{2\epsilon}(\{(i,x_i)\})$.
	If $|I|=2$, then the fact that the liberal version of the tasks does not depend on $\epsilon$ makes the maps equal:
	$\Delta'_\epsilon(\sigma)=
	\{(j,y_j)\mid j\in I,\; \min\{x_i:i\in I\}\leq y_j\leq\max\{x_i:i\in I\}\}=	\Delta_{2\epsilon}(\sigma)$.
	For the case $|I|\geq 3$,
	we prove containment in both directions. 
	First, we consider a simplex $\tau=\{(i,y_i)\mid i\in I\}\in\Delta_\epsilon'(\sigma)$, and show
	$\tau\in\Delta_{2\epsilon}(\sigma)$, i.e., 
	(1)~$\min\{x_i:i\in I\}\leq y_{i'}\leq\max\{x_i:i\in I\}$ 
	for all $i'\in I$,
	and 
	(2)~$|y_i-y_{i'}|\leq 2\epsilon$ for all $i,i'\in I$.
	Since $\tau\in\Delta_\epsilon'(\sigma)$,
	the local task $\Pi_{\tau,\sigma}=(\tau,\Delta_\epsilon(\sigma),\D_{\tau,\sigma})$ is solvable in one round, and let $f:\m{P}^{(1)}(\tau)\to \Delta_\epsilon(\sigma)$
	be a simplicial map solving it.
	As before, for every $v=(i,y_i)\in\tau$, $i\in I$, we have $\D_{\tau,\sigma}(v)=\{v\}$ and $v\in \Delta_\epsilon(\sigma)$, 
	i.e., $f(i,\{(i,y_i)\})= (i,y_i)$, and thus property~(1) holds.
	To check property~(2), let us 
	consider three processes $i,j,k\in I$, and let $\gamma$ be such  that  
	$f(k,\{(i,y_i),(j,y_j),(k,y_k)\})=(k,\gamma)$.
	Then
	\[
	\begin{array}{rcl}
		f(i,\{(i,y_i)\}) & = &(i,y_i),\\
		f(j,\{(j,y_j)\}) & = &(j,y_j),\\
		f(k,\{(i,y_i),(j,y_j),(k,y_k)\}) & = &(k,\gamma).
	\end{array}
	\]
	Since $f$ solves the local task $\Pi_{\tau,\sigma}=(\tau,\Delta_\epsilon(\sigma),\D_{\tau,\sigma})$, we have 
	$|y_i-y_j|
	\leq |y_j-\gamma|+|\gamma-y_i|
	\leq 2\epsilon$.
	It follows that $\tau\in\Delta_{2\epsilon}(\sigma)$. 
	
	Conversely, let us consider $\tau=\{(i,y_i)\mid i\in I\}\in\Delta_{2\epsilon}(\sigma)$,
	and let us define
	$f:\m{P}^{(1)}(\tau)\to \Delta_\epsilon(\sigma)$ by setting, for every $i\in I$ and every set $J\subseteq I$ satisfying $i\in J$,
	\begin{equation}
		f(i,\{y_j : j\in J\})= \Big (i,\min\big\{
		\max\{y_j:j\in J\},
		\min\{y_j:j\in J\}+\epsilon\big\}\Big).
		\label{eq:algo-for-3-and-more-proc}
	\end{equation}
	To show that this map $f$ solves the local task $\Pi_{\tau,\sigma}=(\tau,\Delta_\epsilon(\sigma),\D_{\tau,\sigma})$ in one round, it is sufficient to show that  $f$ solves $\epsilon$-approximate agreement whenever the inputs are at distance at most $2\epsilon$ apart from each other. 
	So, let us 
	consider two processes $i,i'$, and let $J,J'$ be the sets of processes that these two processes see in their immediate snapshot, respectively.
	Assume, w.l.o.g., that $J'\subseteq J$ (this must be the case for outputs of immediate snapshot).
	Let
	$y_{\min}$ and $y_{\max}$ be the minimal and maximal values in $\{y_j\mid j\in J\}$, and let
	$y'_{\min}$ and $y'_{\max}$ be the minimal and maximal values in $\{y'_j\mid j\in J'\}$.
	Note that 
	$y_{\min}\leq y'_{\min}\leq y'_{\max}\leq y_{\max}\leq y_{\min}+2\epsilon$.
	If $y'_{\max}\leq y_{\min}+\epsilon$
	then the facts that
	$f(i,\{y_j : j\in J\})\leq y_{\min}+\epsilon$
	and 
	$f(i,\{y_j : j\in J'\})\leq y'_{\max}$
	imply that both processes output values in $[y_{\min},y_{\min}+\epsilon]$, and we are done.
	Otherwise, 
	$y_{\min}+\epsilon< 
	y'_{\max} \leq 
	y_{\max}$, 
	and so process 	
	$i$ outputs $y_{\min}+\epsilon$,
	while process $i'$ outputs a value in 
	$[y_{\min},y_{\min}+2\epsilon]$, and we are done too.
	
	We conclude that $\Delta'_\epsilon$ and $\Delta_{2\epsilon}$ are identical on simplices of all sizes. 
\end{proof}

We are now ready to prove a lower bound for approximate agreement.

\begin{corollary}\label{lem:approx-agree-without-BB}
	The number of rounds required for solving the $\epsilon$-approximate agreement task  in the wait-free IIS model with $n$ processes is at least 
	$\ceil{\log_3(1/\epsilon)}$ if $n=2$, and 
	$\ceil{\log_2(1/\epsilon)}$ if $n>2$.
\end{corollary}

\begin{proof}
	Fix $\epsilon$ and $n$, 
	and let $\Pi$ denote the $\epsilon$-approximate agreement task.
	Consider an algorithm solving $\epsilon$-approximate agreement in $t$ rounds.
	If $n=2$, then, by Theorem~\ref{thm:speedup} and Claim~\ref{claim:epprox-agree-n=2-one-round},
	there is an algorithm solving 
	$(3\epsilon)$-approximate agreement in $t-1$ rounds. By repeating this argument $t$ times, we get a $(3^t\epsilon)$-approximate agreement algorithm running in $0$~rounds.
	By Claim~\ref{claim:epprox-agree-zero-rounds},
	we have $3^t\epsilon\geq 1$,
	which immediately implies the lower bound.
	If $n\geq 3$, first note that the alleged algorithm 
	also solves the liberal version of $\epsilon$-approximate agreement in $t$ rounds.
	By the same proof as above, but using  Claim~\ref{claim:epprox-agree-n>2-one-round} instead of Claim~\ref{claim:epprox-agree-n=2-one-round},
	we get
	that there is a $0$-round algorithm solving the liberal version of $(2^t\epsilon)$-approximate agreement,
	implying 
	$2^t\epsilon\geq 1$,
	and the lower bound follows.
\end{proof}

\paragraph{Remark.}

The bounds in Corollary~\ref{lem:approx-agree-without-BB} are known to be tight~\cite{HoestS06}. 
In fact, the tightness of our lower bounds is also directly implied by our proofs, as specified by Equations~\ref{eq:algo-for-2-proc} and~\ref{eq:algo-for-3-and-more-proc}. The former is for two processes, and essentially divides the distances by~3 at each round. The latter is for three or more processes, and essentially halves the distances at each round.

\subsection{Approximate Agreement in the Wait-Free IIS Model Augmented with \TS}

We now turn our attention to the approximate agreement task in an augmented version of the previous model,
namely, we consider a model of iterated immediate snapshot augmented with \TS{}.
It is known that the consensus number of \TS{} is~$2$, 
and, as explained in the previous section, 
in the model of immediate snapshot augmented with \TS,
consensus among two processes can be solved in a single round.
This immediately implies that the simpler task of approximate agreement among two processes is also solvable in a single round in this model.
A more surprising result asserts that with three or more processes, the \TS{} object does not accelerate much the solution of approximate agreement. Indeed, 
using the round reduction theorem for augmented models (Theorem~\ref{thm:speedup-extended}), we show that it still takes at least $\lceil\log_2 1/\epsilon\rceil$ rounds to solve $\epsilon$-approximate agreement.

\begin{theorem}\label{theo:approx-agree-with-ts}
	The number of rounds required for solving the $\epsilon$-approximate agreement task   in the wait-free IIS model augmented with \TS{} among $n\geq 3$ processes is at least $\ceil{\log_2(1/\epsilon)}$.
\end{theorem}

\begin{proof}
	To establish the theorem,  we first prove an analogue of Claim~\ref{claim:epprox-agree-n>2-one-round},
	for the model of iterated immediate snapshot augmented with \TS. 
	
	\begin{claim}\label{claim:epprox-agree-n>2-one-round-ts}
		If $\Pi=(\m{I},\m{O},\Delta)$ is the liberal version of $\epsilon$-approximate agreement with $n\ge 3$ processes,
		then the closure of $\Pi$ with respect to the wait-free IIS model augmented with \TS\/ is the liberal version of $(2\epsilon)$-approximate agreement.
	\end{claim}
	
	\begin{proof}[Proof of claim]
		Let $\Pi_\epsilon=(\m{I},\m{O}_\epsilon,\Delta_\epsilon)$ be the liberal version of $\epsilon$-approximate agreement for $n\geq 3$ processes, and let $\cl_M(\Pi_\epsilon)=(\m{I},\m{O}'_\epsilon,\Delta'_\epsilon)$ be the closure of $\Pi_\epsilon$ in the considered model~$M$. 
		Let us denote the liberal version of $(2\epsilon)$-approximate agreement by $\Pi_{2\epsilon}=(\m{I},\m{O}_{2\epsilon},\Delta_{2\epsilon})$. 
		We want to show that, for every $\sigma\in\m{I}$, $\Delta'_\epsilon(\sigma)=\Delta_{2\epsilon}(\sigma)$. If $|\sigma|\leq 2$, this equality holds exactly by the same reasoning as in the previous proofs
		(Claim~\ref{claim:epprox-agree-n=2-one-round} for $|\sigma|= 1$ and 
		Claim~\ref{claim:epprox-agree-n>2-one-round} for $|\sigma|= 2$),
		so we focus on the case where $|\sigma|\geq 3$.
		
		Consider an input simplex $\sigma=\{(i,x_i)\mid i\in I\}\in\m{I}$, for some set $I\subseteq[n]$, $|I|\geq 3$.
		We consider a simplex $\tau=\{(i,y_i)\mid i\in I\}\in\Delta_\epsilon'(\sigma)$, and show
		$\tau\in\Delta_{2\epsilon}(\sigma)$, i.e., 
		(1)~$\min\{x_i:i\in I\}\leq y_i\leq\max\{x_i:i\in I\}$ 
		for all $i\in I$,
		and 
		(2)~$|y_i-y_{j}|\leq 2\epsilon$ for all $i,j\in I$.
		Since $\tau\in\Delta_\epsilon'(\sigma)$,
		the local task $\Pi_{\tau,\sigma}=(\tau,\Delta_\epsilon(\sigma),\D_{\tau,\sigma})$ is solvable in one round, and let $f:\m{P}^{(1)}(\tau)\to \Delta_\epsilon(\sigma)$
		be a simplicial map solving it.
		As before, 
		$\tau\subseteq V(\Delta_\epsilon(\sigma))$
		implies property~(1).
		To check property~(2),
		let us focus on the simplices of $\m{P}^{(1)}(\sigma)$ considered in the proof of Corollary~\ref{cor: consensus impossibility with ts} (see also Figure~\ref{fig:n-proc-consensus-not-solvable}):
		\begin{align*}
			\rho_{i,j,k}=\Big\{
			\big(i,(1,\{(i,y_i)\})\big),  \big(j,(0,\{(i,y_i),(j,y_j)\})\big),     
			\big(k,(0,\{(i,y_i),(j,y_j),(k,y_k)\})\big)
			\Big\},\\
			\rho_{j,i,k}=\Big\{
			\big(j,(1,\{(j,y_j)\})\big),  \big(i,(0,\{(i,y_i),(j,y_j)\})\big),     
			\big(k,(0,\{(i,y_i),(j,y_j),(k,y_k)\})\big)
			\Big\}.
		\end{align*}
		Let  $(k,\gamma)=f\big(k,(0,\{(i,y_i),(j,y_j),(k,y_k)\})\big)$ be the output at vertex $\big(k,(0,\{(i,y_i),(j,y_j),(k,y_k)\})\big)$.
		Since $f$ agrees with $\Delta_{\tau,\sigma}$, we have that $f(\rho_{i,j,k})\subseteq\proj_{\{i,j,k\}}(\Delta_\epsilon(\sigma))$, and the definition of $\Delta_\epsilon$ implies $|y_i-\gamma|\leq\epsilon$.
		A similar argument on $f(\rho_{j,i,k})$ implies $|y_j-\gamma|\leq\epsilon$,
		and the triangle inequality hence implies $|y_i-y_j|\leq 2\epsilon$, as desired.
		
		Conversely, let us consider $\tau=\{(i,y_i)\mid i\in I\}\in\Delta_{2\epsilon}(\sigma)$,
		and let us define
		$f:\m{P}^{(1)}(\tau)\to \Delta_\epsilon(\sigma)$
		in a very similar way to the one in the proof of Claim~\ref{claim:epprox-agree-n>2-one-round},
		ignoring the output $t_i$ of the \TS{} object.
		That is, for every $i\in I$ and every set $J\subseteq I$ satisfying $i\in J$, we set
		\[
		f\big(i,(t_i,\{(j,y_j)\mid j\in J\})\big)=\Big (i,\min\big\{
		\max\{y_j:j\in J\},
		\min\{y_j:j\in J\}+\epsilon\big\}\Big).
		\]
		As for the proof of Claim~\ref{claim:epprox-agree-n>2-one-round}, this map $f$ solves $\epsilon$-approximate agreement whenever $|y_i-y_{j}|\leq 2\epsilon$ for every $i,j\in I$, and therefore it solves the local task $\Pi_{\tau,\sigma}=(\tau,\Delta_\epsilon(\sigma),\D_{\tau,\sigma})$ in one round. It follows that $\tau\in\Delta'_\epsilon(\sigma)$. 
		
		The two maps $\Delta'_\epsilon$ and $\Delta_{2\epsilon}$ are therefore identical, which completes the proof of Claim~\ref{claim:epprox-agree-n>2-one-round-ts}. 
	\end{proof}
	
	We complete the proof of Theorem~\ref{theo:approx-agree-with-ts} by the same reasoning as in the proof of 
	Corollary~\ref{lem:approx-agree-without-BB},
	but using 
	Theorem~\ref{thm:speedup-extended} and Claim~\ref{claim:epprox-agree-n>2-one-round-ts}.
	Let us fix $\epsilon>0$ and $n\geq 3$.
	Consider an algorithm solving $\epsilon$-approximate agreement in $t$ rounds---%
	this algorithm also solves the liberal version of
	$\epsilon$-approximate agreement in $t$ rounds.
	By
	Theorem~\ref{thm:speedup-extended} and Claim~\ref{claim:epprox-agree-n>2-one-round-ts},
	there is an algorithm solving 
	the liberal version of $(2\epsilon)$-approximate agreement in $t-1$ rounds,
	and repeating this argument for $t$ times implies an algorithm solving the liberal version of  $(2^t\epsilon)$-approximate agreement algorithm in $0$~rounds. 
	Now, note that Claim~\ref{claim:epprox-agree-zero-rounds} still holds, since the \TS{} object is not used in a zero round algorithm, 
	thus we have $2^t\epsilon\geq 1$,
	which immediately implies the lower bound.
\end{proof}

\subsection{Approximate Agreement in the Wait-Free IIS Model  with Binary Consensus}

\begin{figure}[tb]
	\centering
	\includegraphics[width=8cm]{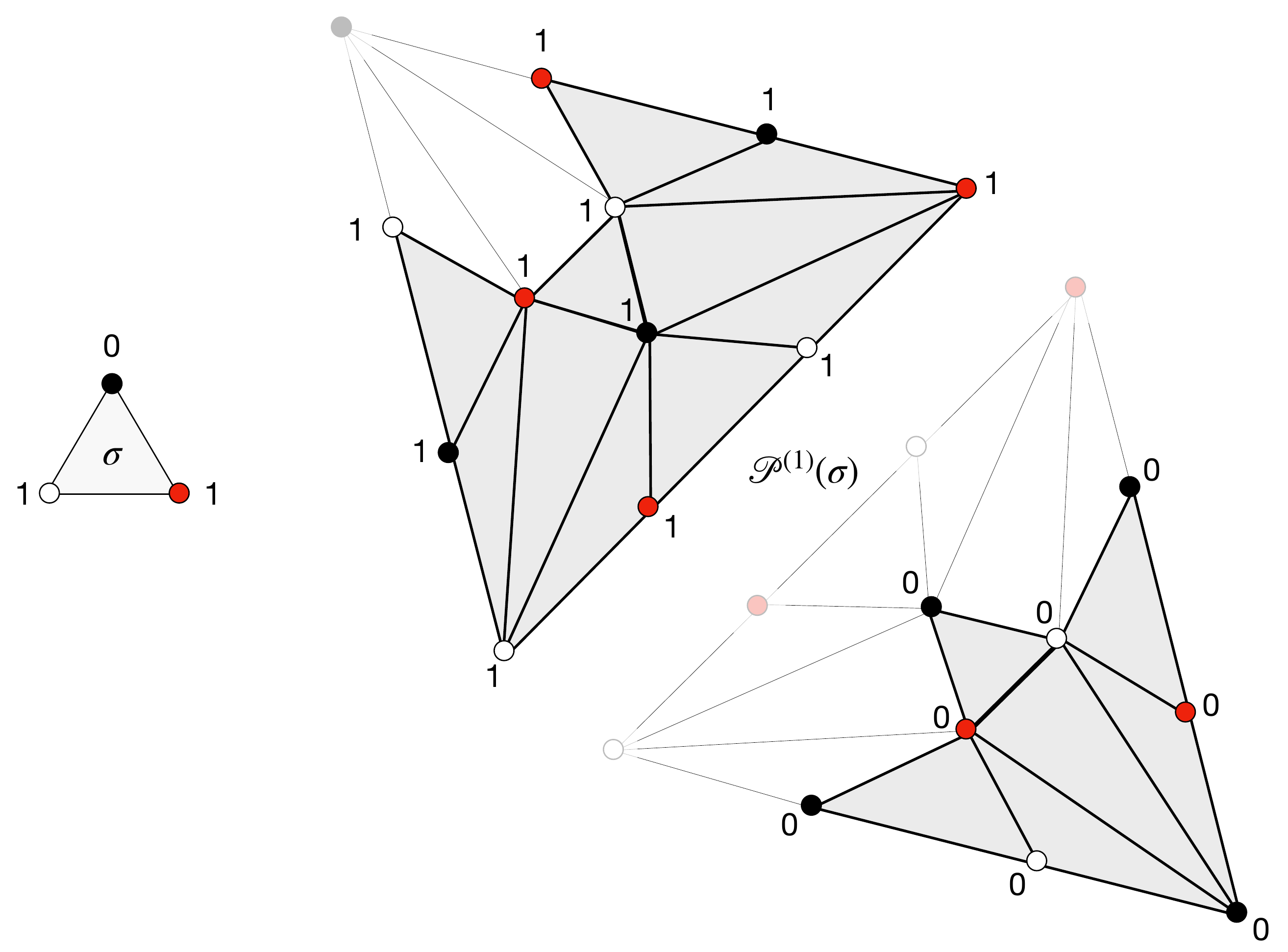}
	\caption{\sl The 1-round protocol complex using binary consensus when the black process calls consensus with input~0, and the other two processes call consensus with input~1.}
	\label{fig:prot-complex-with-bin-consensus}
\end{figure}

In this section, we consider again the approximate agreement task, but in the model of iterated immediate snapshot augmented with binary consensus.
Clearly, in this model, $\epsilon$-approximate agreement is solvable for every $\epsilon$ whenever $n=2$. Figure~\ref{fig:prot-complex-with-bin-consensus} displays the $1$-round protocol complex $\m{P}^{(1)}(\sigma)$ for a $2$-dimensional simplex~$\sigma$.
Here, the vertices of the input simplex~$\sigma$ are labeled with the values with which the processes call the binary consensus object,
and the vertices of the protocol complex $\m{P}^{(1)}(\sigma)$ are labeled with the output values of this object.
The protocol complex $\m{P}^{(1)}(\sigma)$ can thus be viewed as two copies of the chromatic subdivision of~$\sigma$, one for which consensus is~$0$, and the other for which consensus is~$1$. However, in each copy, some simplices are removed, depending of the input given to the consensus object. For instance, in the figure, the black process calls the object with input~$0$, so it must also output~$0$ in a solo execution,
and the vertex where it outputs $1$ in a solo execution is removed. 
Similarly, the white and red processes call the object with input~$1$, so all executions in which only these processes participate result in both processes outputting~$1$. 

We are aware of two techniques for solving approximate agreement using binary consensus, both consist in agreeing on a value one bit at a time.
One technique is solving (exact) consensus in $\lceil\log_2 n\rceil$ rounds, by agreeing on the ID of one of the participating processes, and then deciding on its input.
The other techniques takes $\lceil \log_2 1/\epsilon\rceil$ rounds, and goes as follows. 
At round $r$, each process writes its current value, then evokes the binary consensus object with the $r$-th bit of its current value (starting from the most significant bit), and, finally, updates its current value to one of the proposed values whose $r$-th bit is equal to the output of the consensus object.
Note that the call to the binary consensus object in the first type of algorithms depends only on the process ID, and not on its input or current value. Instead, in the second type of algorithms, the call to the binary consensus object depends solely on the value and not on the process~ID. 
In this section, we give a lower bound that  applies to algorithms of the first type, i.e., algorithms where the call of a process to the binary consensus object depends solely on its ID and round number.

\begin{theorem}\label{thm:approx-agree-with-bc}
	The number of rounds required by $n\geq 3$ processes
	for solving the $\epsilon$-approximate agreement task in the wait-free IIS model augmented with a binary consensus object called with inputs depending only on the process~IDs and round numbers
	is at least $\min\{\ceil{\log_2 1/\epsilon},\ceil{\log_2 n}-1\}$.
\end{theorem}

\begin{proof}
	Let $n\geq 3$.
	Let  $M$ be the iterated immediate snapshot model augmented with a binary consensus object, 
	with the restriction that the input value used by every process~$i\in[n]$ to call the binary consensus object at each round~$r\geq 1$ may depend only on~$i$ and~$r$. Under these assumptions, we can strengthen the statement of Theorem~\ref{thm:speedup-extended} as follows. Let $\beta:[n]\to\{0,1\}$ be a function. We define the closure of a task~$\Pi=(\m{I},\m{O},\Delta)$ with respect to~$\beta$, denoted by $\cl_M(\Pi|\beta)$, as the task $\Pi=(\m{I},\m{O}',\Delta')$ where, for every simplex $\sigma\in\m{I}$, and for  every chromatic set $\tau\subseteq V(\Delta(\sigma))$, we set 
	$\tau\in \Delta'(\sigma)$ if and only if the local task $\Pi_{\tau,\sigma}=(\tau,\Delta(\sigma),\Delta_{\tau,\sigma})$ is solvable in one round by an algorithm in which, for every $i\in [n]$, process~$i$ calls the binary consensus object with input~$\beta(i)$. The following claim is an analog of Theorem~\ref{thm:speedup-extended} for the closure task with respect to functions~$\beta$. 
	
	\begin{claim}\label{claim:speedup-extended-beta}
		For every $t\geq 1$, if  a task $\Pi$ is solvable in $t$ rounds in the model~$M$, then there exists a function $\beta:[n]\to\{0,1\}$ such that $\cl_M(\Pi|\beta)$ is solvable  in $t-1$ rounds in~$M$.
	\end{claim}
	
	\begin{proof}[Proof of claim]
		The proof is identical to the proof of Theorem~\ref{thm:speedup-extended}, by noticing that, during its $t$-th round, every process~$i$ running the algorithm calls the binary consensus object with input $\alpha(i,t)$, which is independent of the view and depends only on the process~ID and the round. 
		The desired function $\beta$ is therefore merely defined as $\beta(i)=\alpha(i,t)$ for every $i\in[n]$. 
	\end{proof}
	
	By fixing the function $\beta$ used by the processes for calling the binary consensus object, we will be able to show that if the original task is $\epsilon$-approximate agreement, then the closure $\cl_M(\Pi|\beta)$ with respect to~$\beta$ is actually $(2\epsilon)$-approximate agreement whenever only some (large) group of processes participate. More specifically, let $S$ be the largest of the two sets $\beta^{-1}(0)$ and $\beta^{-1}(1)$, with $S=\beta^{-1}(0)$ if they are of equal sizes,
	and note that $|S|\geq n/2$. 
	Intuitively, when only the processes in~$S$ participate, they take no benefit of the consensus object, to which they all input the same value~$a$, and from which they all get the same output~$b=a$. 
	Hence, when only processes in~$S$ participate the closure of $\epsilon$-approximate agreement with respect to~$\beta$ is $(2\epsilon)$-approximate agreement, as is the case without binary consensus object. 
	We then repeat the same argument on $S$, by identifying a set $S'\subseteq S$ that may result from another function, $\beta'$, of size at least $n/4$, such that, whenever only the processes in~$S'$ participate, the closure with respect to~$\beta'$ of the closure of $\epsilon$-approximate agreement with respect to~$\beta$ is $(4\epsilon)$-approximate agreement. 
	At each iteration of this reasoning, we halve the number of processes, and double the precision parameter of the approximate agreement task, until 
	either the number of processes becomes very small
	or the precision parameter becomes very large,
	in which cases the task becomes easier.
	Therefore, the number of iterations of this reasoning is roughly  $\min\{\log_2 1/\epsilon,\log_2n\}$. 
	The exact figures are slightly more complicate because, for two processes, $\epsilon$-approximate agreement does not require at least $\lceil\log_21/\epsilon\rceil$ rounds but $\lceil\log_31/\epsilon\rceil$ rounds in absence of a binary consensus object, and just a single round with such an object.
	
	The following claim is at the core of the proof. 
	
	\begin{claim}\label{claim:epprox-agree-n>4-one-round-bc}
		Let $\Pi=(\m{I},\m{O},\Delta)$ be the  liberal version of $\epsilon$-approximate agreement task on a set
		$S\subseteq [n]$ of $|S|\geq 5$ processes,
		and let $\beta:S\to\{0,1\}$ be any function. 
		Then there is a subset $S'\subseteq S$ of processes, with $|S'|\geq|S|/2$,
		such that  $\cl_{M}(\Pi|\beta)$ is the liberal version of $(2\epsilon)$-approximate agreement task whenever only the processes of $S'$ participate.
	\end{claim}
	
	\begin{proof}[Proof of claim]
		Let $S'\subseteq S$ be the largest of the sets $\beta^{-1}(0)$ and $\beta^{-1}(1)$ (and $\beta^{-1}(0)$ in case of equal sizes). 
		We henceforth assume, without lost of generality, that $S'=\beta^{-1}(0)$. Note that $|S'|\geq 3$, as $|S|\geq 5$.
		Let $\Pi_\epsilon=(\m{I},\m{O}_\epsilon,\Delta_\epsilon)$ be the liberal version of $\epsilon$-approximate agreement task on $S$, and let $\cl_M(\Pi_\epsilon|\beta)=(\m{I},\m{O}'_\epsilon,\Delta'_\epsilon)$. 
		Let us denote the liberal version of $(2\epsilon)$-approximate agreement on $S$ by $\Pi_{2\epsilon}=(\m{I},\m{O}_{2\epsilon},\Delta_{2\epsilon})$. 
		We want to show that, restricted to the processes in~$S'$, the two tasks $\Pi_{2\epsilon}$ and $\cl_M(\Pi_\epsilon|\beta)$ are  identical. Formally, we want to show that, 
		for every $\sigma\in \m{I}$ with $\ID(\sigma)\subseteq S'$, $\Delta'_\epsilon(\sigma)=\Delta_{2\epsilon}(\sigma)$.
		If $\sigma=\{v\}$ is a single vertex, then $\Delta_{2\epsilon}(\{v\})=\{v\}$, and $\Delta'_\epsilon(\{v\})=\{v\}$, and thus $\Delta'_\epsilon(\sigma)=\Delta_{2\epsilon}(\sigma)$ as desired. 
		If $|\sigma|=2$
		then the fact that the liberal version of the tasks does not depend on $\epsilon$ makes the maps equal:
		$\Delta'_\epsilon(\sigma)=
		\{(j,y_j)\mid j\in I,\; \min\{x_i:i\in I\}\leq y_j\leq\max\{x_i:i\in I\}\}=	\Delta_{2\epsilon}(\sigma)$.
		Let us now consider an input simplex $\sigma=\{(i,x_i)\mid i\in I\}\in\m{I}$, for some set $I\subseteq S'$, with $|I|\geq 3$.
		
		First, we show that $\Delta_\epsilon'(\sigma)\subseteq \Delta_{2\epsilon}(\sigma)$. Let $\tau=\{(i,y_i)\mid i\in I\}\in\Delta_\epsilon'(\sigma)$. We want to show that 
		$\tau\in\Delta_{2\epsilon}(\sigma)$, i.e., 
		(1)~$\min\{x_i:i\in I\}\leq y_i\leq\max\{x_i:i\in I\}$ 
		for all $i\in I$,
		and 
		(2)~$|y_i-y_{j}|\leq 2\epsilon$ for all $i,j\in I$.
		Since $\tau\in\Delta_\epsilon'(\sigma)$,
		the local task $\Pi_{\tau,\sigma}=(\tau,\Delta_\epsilon(\sigma),\D_{\tau,\sigma})$ is solvable in one round, using an algorithm in which each process $i\in I$ calls the binary consensus object with input $\beta(i)=0$.  
		Let $f:\m{P}^{(1)}(\tau)\to \Delta_\epsilon(\sigma)$
		be a simplicial map solving the local task $\Pi_{\tau,\sigma}$ where $f:\m{P}^{(1)}(\tau)$ is the 1-round protocol complex in model~$M$ starting from~$\tau$ with input~0 to the binary consensus object.
		Since $\tau\subseteq V(\Delta_\epsilon(\sigma))$, property~(1) holds.
		To check property~(2), note that since all processes in $I\subseteq S'$ call the binary consensus object with the same input value~0, they necessarily all get the same value~0 as output.
		This allows us to follow arguments similar to the ones used in the proof of Claim~\ref{claim:epprox-agree-n>2-one-round}.
		Let us consider three processes $i,j,k\in I$, and let us define $\gamma$ by   
		$
		f(k,(0,\{(i,y_i),(j,y_j),(k,y_k)\}))=(k,\gamma),
		$
		where the~0 in this definition is the output of binary consensus at process~$k$.
		Since $f$ solves the local task~$\Pi_{\tau,\sigma}$, we have 
		$
		f(i,(0,\{(i,y_i)\})) = (i,y_i) \; \mbox{and}\;  f(j,(0,\{(j,y_j)\})) = (j,y_j),
		$
		where, again, the~$0$ values in these equalities are the outputs of binary consensus at processes~$i$ and~$j$, respectively. 
		Since $f$ solves the local task $\Pi_{\tau,\sigma}=(\tau,\Delta_\epsilon(\sigma),\D_{\tau,\sigma})$, we have 
		\[
		\big\{ f(i,(0,\{(i,y_i)\})), \; f(k,(0,\{(i,y_i),(j,y_j),(k,y_k)\})) \big\}\in \Delta_\epsilon(\sigma)
		\]
		and 
		\[
		\big\{ f(j,(0,\{(j,y_j)\})), \; f(k,(0,\{(i,y_i),(j,y_j),(k,y_k)\})) \big\}\in \Delta_\epsilon(\sigma).
		\]
		It follows that 	$|y_i-y_j|
		\leq |y_j-\gamma|+|\gamma-y_i|
		\leq 2\epsilon$, and therefore  $\tau\in\Delta_{2\epsilon}(\sigma)$. 
		
		Conversely, we show that $\Delta_{2\epsilon}(\sigma)\subseteq \Delta_\epsilon'(\sigma)$. Let $\tau=\{(i,y_i)\mid i\in I\}\in\Delta_{2\epsilon}(\sigma)$. We show that the local task $\Pi_{\tau,\sigma}=(\tau,\Delta_\epsilon(\sigma),\D_{\tau,\sigma})$ is solvable in one round, using an algorithm in which each process $i\in I$ calls the binary consensus object with input $\beta(i)=0$.  Since $I\subseteq S'$, the output of the object is necessarily~$0$ at every process~$i\in I$. 
		Let
		$f:\m{P}^{(1)}(\tau)\to \Delta_\epsilon(\sigma)$ be as follows. For every $i\in I$, and every set $J\subseteq I$ satisfying $i\in J$, we set 
		\[
		f(i,(0,\{y_j : j\in J\}))= \Big (i,\min\big\{
		\max\{y_j:j\in J\},
		\min\{y_j:j\in J\}+\epsilon\big\}\Big).
		\]
		This map $f$ solves $\epsilon$-approximate agreement whenever $|y_i-y_j|\leq 2\epsilon$ for any $i,j\in I$, and therefore it solves the local task $\Pi_{\tau,\sigma}=(\tau,\Delta_\epsilon(\sigma),\D_{\tau,\sigma})$ in one round. It follows that $\tau\in\Delta'_\epsilon(\sigma)$.  This completes the proof of Claim~\ref{claim:epprox-agree-n>4-one-round-bc}. 
	\end{proof}
	
	We now have all the ingredients for establishing the lower bound. For this purpose, let $t=\min\{\lceil\log_2n\rceil-1, \lceil\log_2 1/\epsilon\rceil\}$, and let us assume, for the purpose of contradiction, that there exists an algorithm solving the liberal version of $\epsilon$-approximate agreement in $t-1$ rounds. By Claims~\ref{claim:speedup-extended-beta} and~\ref{claim:epprox-agree-n>4-one-round-bc}, this implies the existence of a $(t-2)$-round algorithm solving the liberal version of  $(2\epsilon)$-approximate agreement among a set $S_1\subseteq [n]$ of processes, with $|S_1|\geq n/2$.  By iterating the application of Claims~\ref{claim:speedup-extended-beta} and~\ref{claim:epprox-agree-n>4-one-round-bc}, we eventually get a 0-round algorithm solving the liberal version of  $(2^{t-1}\epsilon)$-approximate agreement among a set $S_{t-1}\subseteq [n]$ of processes, with $|S_{t-1}|\geq n/2^{t-1}$.  Since $t\leq \lceil\log_2n\rceil-1$, using the fact that $\lceil\log_2n\rceil<\log_2n+1$, we get $|S_{t-1}|> \frac{n}{2^{\log_2n-1}}$, and thus $|S_{t-1}|\geq 3$.
	Also, since $t\leq \lceil\log_2 1/\epsilon\rceil$, using similarly the fact 
	that,
	$\lceil\log_21/\epsilon\rceil<\log_21/\epsilon+1$, we get that $2^{t-1}\epsilon<1$. Therefore, we get a 0-round algorithm solving the liberal version of  $\epsilon'$-approximate agreement among a set of at least three processes, where $\epsilon'<1$. This is a contradiction with Claim~\ref{claim:epprox-agree-zero-rounds}, which still holds in the context of the theorem since, in $0$~rounds, the binary consensus object is not used. Therefore, the liberal version of  $\epsilon$-approximate agreement requires at least $t$~rounds to be solved, and so does the (standard version of) $\epsilon$-approximate agreement. 
\end{proof}

\section{Conclusion}

Our results open many interesting questions. 
Our asynchronous speedup theorem is inspired by an analogous speedup technique~\cite{Balliu0HORS19,Brandt19} for the \textsf{LOCAL} model, and we believe that the two forms of speedup theorems
shed light on the differences between asynchronous computation and computation in the \textsf{LOCAL} model. 
In particular, it seems that while in the \textsf{LOCAL} model it is possible to state a generic ``if and only if''
form of a speedup theorem, this is not the case for asynchronous computation in the wait-free model.

We have established our speedup theorem in iterated models.
It would be interesting to extend it to non-iterated models, in which a register can be used multiple times.
This seems especially interesting, since while the iterated and the non-iterated models are known to be equivalent with respect to task solvability~\cite{BorowskyG97,GafniR10,BouzidGK14,ImbsRV15} (even when enriched with some more powerful objects), 
they are not known to be equivalent in terms of time complexity. 
Moreover, our speedup theorem assumes models allowing processes to run solo, as in standard wait-free models, 
and it would be interesting to study extensions to affine models~\cite{KRH18} without this property, such as $t$-resilient models for $t<n-1$. 

We focused on some objects of consensus-number~$2$, such as
\TS\/, and showed that they are useless for solving approximate agreement faster.
We would be interested in extending this result to every consensus-number-2 object.
However, we showed that even objects with
consensus number $\infty$, such as binary consensus, are not useful for solving  $\epsilon$-approximate agreement (at least when $n\geq 1/\epsilon$) faster,
so perhaps the answer is independent of the consensus hierarchy.
And finally, it would be interesting to use the  speedup theorem for problems other than consensus and approximate agreement.

\bibliographystyle{ACM-Reference-Format}
\bibliography{speedup}


\newpage
\appendix
\centerline{\Large \bf A P P E N D I X}

\section{Model}
\label{app:model}


\subsection{Element of Algebraic Topology}

Recall that a \emph{complex} is a collection~$\m{K}$ of non-empty sets, closed under inclusion, i.e., if $\sigma\in\m{K}$ then, for every non-empty set $\sigma'\subseteq \sigma$, $\sigma'\in\m{K}$. Every set in $\m{K}$ is called a \emph{simplex}. A subset of a simplex is called a \emph{face}, and a \emph{facet} of~$\m{K}$ is a face that is maximal for inclusion in~$\m{K}$. The \emph{dimension} of a simplex $\sigma$ is $|\sigma|-1$, where $|\sigma|$ denotes the cardinality of~$\sigma$. The dimension of a complex is the maximal dimension of its facets. A complex in which all facets are of the same dimension is called \emph{pure}. The  \emph{vertices} of $\m{K}$ are all simplices with a single element (i.e., of dimension~0). The set of vertices of a complex~$\m{K}$ are denoted by~$V(\m{K})$. 

All complexes in this paper are \emph{chromatic}, i.e., every vertex is a pair $v=(i,x)$ where $i\in[n]=\{1,\dots,n\}$ for some $n\geq 1$ is the \emph{color} of~$v$, and $x$ is some value (e.g., an input value, an output value, or some value corresponding to a set of data acquired after some computation). Moreover, in a chromatic complex, a ``color''~$i$ must appear at most once in every simplex, that is, a simplex is of the form $\sigma=\{(i,v_i):i\in I\}$ for some non-empty set $I\subseteq [n]$.  
Appendix~\ref{app:3topologies} contains figures depicting some chromatic simplicial complexes; these specific simplices are discussed later.

A \emph{simplicial map} from a complex $\m{K}$ to a complex $\m{K}'$ is a map $f:V(\m{K})\to V(\m{K}')$ preserving simplices, that is, for every simplex  $\sigma$, the set $f(\sigma)$ is a simplex of~$\m{K}'$. All simplicial maps considered in this paper are chromatic, that is, they preserve the colors of the vertices, i.e., for every $(i,x)\in V(\m{K})$, $f(i,x)=(i,y)\in V(\m{K}')$. So, given a chromatic simplicial map from a chromatic complex $\m{K}$ to a chromatic complex $\m{K}'$, for every $\sigma=\{(i,x_i):i\in I\}\in\m{K}$, the set $f(\sigma)=\{f(i,x_i):i\in I\}$ is a chromatic simplex of $\m{K}$ on the same set of colors.  Since simplicial maps sends simplices to simplices, we often write $f:\m{K}\to \m{K}'$ even if~$f$ is actually defined on vertices. 

A \emph{carrier map} from a chromatic complex $\m{K}$ to a chromatic complex $\m{K}'$ is a map $\Delta:\m{K}\to 2^{\m{K}'}$ which maps every simplex $\sigma\in\m{K}$ to a pure sub-complex of~$\m{K}'$ with same dimension and same colors as~$\sigma$ such that, for every $\sigma'\subseteq\sigma$,  $\Delta(\sigma')\subseteq \Delta(\sigma)$, i.e., $\Delta(\sigma')$ is a subcomplex of $\Delta(\sigma)$. 

Any simplex $\sigma$ of a (chromatic) complex~$\m{K}$ can also be viewed as a complex $\bar{\sigma}$ whose simplices are all the faces of~$\sigma$. To avoid overloading the notations, we will omit the bar over $\sigma$ when it is clear from the context whether we consider $\sigma$ as a simplex or as a complex. The same holds for a collection of simplices of a (chromatic) complex~$\m{K}$.

\paragraph{Notations.}

Let $\sigma=\{(i,x_i):i\in I\}$ be a simplex. We denote by $\ID(\sigma)$ the set of colors in~$\sigma$, i.e., $\ID(\sigma)=I$. Indeed, in the following, the color of a vertex is actually the \emph{identity} of a process. For every non-empty set $J\subseteq I$, we denote by $\proj_J(\sigma)$ the simplex $\{(i,x_i):i\in J\}$, i.e., $\ID(\proj_J(\sigma))=J$. That is, $\proj_J(\sigma)$ is the \emph{projection} of $\sigma$ resulting from considering only vertices with colors in~$J$. 

\subsection{Tasks}

We consider distributed systems with  $n\geq 2$ processes, labeled by distinct integers from~1 to~$n$. Every process~$i$ initially know its identity~$i$ as well as the total number~$n$ of processes in the system. A \emph{task} for $n$ processes is a triple $(\m{I},\m{O},\Delta)$ where $\m{I}$ and $\m{O}$ are $(n-1)$-dimensional  complexes, respectively called \emph{input} and \emph{output} complexes, and $\Delta:\m{I}\to 2^{\m{O}}$ is an input-output specification. Every simplex $\sigma=\{(i,x_i):i\in I\}$ of $\m{I}$ defines a legal  input state corresponding to the scenario in which, for every $i\in I$, process~$i$ starts with input value~$x_i$. Similarly, every simplex $\tau={\{(i,y_i):i\in I\}}$ of $\m{O}$ defines a legal output state corresponding to the scenario in which, for every $i\in I$, process~$i$ outputs the value~$y_i$. The map~$\Delta$ is an input-output relation specifying, for every input state $\sigma\in\m{I}$, the set of output states~$\tau\in\m{O}$ with $\ID(\tau)=\ID(\sigma)$ that are legal with respect to~$\sigma$. That is, assuming that only the processes in~$\ID(\sigma)$ participate to the computation (the set of participating processes is not known a priori to the processes in~$\sigma$), there processes are allow to output any simplex $\tau\in\Delta(\sigma)$. It is usually assumed that~$\Delta$ is a carrier map, but in this paper we do not enforce this requirement into the definition. 

%

\subsection{Computational Model}

We consider asynchronous computing in the standard \emph{read-write shared memory} distributed computing model~\cite{AWbook}, focusing on iterated models~\cite{iterated2010}, 
following the topology approach to distributed computing~\cite{bookHerlihyKR2013}.

\subsubsection{Read-Write Shared Memory}

A \emph{process} is a deterministic (infinite) state machine. We consider a collection of $n\geq 2$ \emph{processes} that exchange information by accessing to a shared memory. Processes are labeled from~1 to~$n$, and process~$i$ is aware of its label, called \emph{identifier}~(ID), as well as of the size~$n$ of the system. The shared memory has $n$ distinct single-writer multiple-reader (SWMR) \emph{registers} $R[1],\dots,R[n]$. Each process~$i\in[n]$ can store data in the  shared memory by \emph{writing} in its register~$R[i]$, and no other processes can write in~$R[i]$. On the other hand, for every $i\in[n]$, $R[i]$ can be \emph{read} by any process~$j\in [n]$. 

Since we are focussing on lower bounds, we assume no limits on the computational power of the processes, and on the size of their private memories. Similarly, we assume no limits to the size of the registers in the shared memory. In particular, we assume that no information is ever lost once written in the shared memory.  Finally,  we assume no limits on the amount of information that can be transferred from and to the memory by the processes. In particular, when a process writes in its register, we assume that it  writes the whole content of its private memory. Similarly when it reads a register, we assume that it reads the whole content of that register. 
These assumptions are referred to as the \emph{full-information} model~\cite{bookHerlihyKR2013}. 

Also, since our objective is the design of lower bounds, we assume that each communication, i.e., writing to, or reading from a register, is \emph{atomic}. Similarly, each sequence of internal computation performed by a process between two successive accesses to the shared memory is supposed to be atomic. However, the overall system is supposed to be \emph{asynchronous}. That is, the time at which each atomic operation of a process is executed is arbitrary and unpredictable, apart from the fact that the atomic operations executed by a same process are executed in order.

\subsubsection{Iterated Models}

We adopt a standard approach for designing robust algorithms in shared-memory systems involving processes subject to crash failures, by focussing on algorithms with a generic form, using the  \emph{collect} instruction. 

\paragraph{Collect:}

For every~$i\in [n]$, the  \emph{collect} instruction performed by process~$i$ results in this process reading all registers~$R[j]$ sequentially, for $j=1,\dots,n$, in arbitrary order.

%

Such a generic form of algorithm is 
Algorithm~\ref{alg:generic-model},  for the \emph{iterated} model, in which the shared-memory is organized in arrays $M_r$, $r\geq 1$, of $n$ SWMR registers (one per process) and each round $r$ of the protocol is executed on a the array~$M_r$. At each round, each process~$i$ updates its so-called \emph{view}~$V_i$ by collecting the views of the other processes.  Initially, the view~$V_i$ of process~$i$ is reduced to its own input~$x_i$. Then, a sequence of $t$ rounds are performed, for some $t\geq 0$, where a round is defined as a sequence of two consecutive write and collect operations.  At the end of round~1, the view~$V_i$ of process~$i$ is a set of pairs $(j,x_j)$. It must be the case that  $(i,x_i)\in V_i$ as the collect instruction is performed after the write instruction. For $j\neq i$, it may or may not be the case that $(j,x_j)\in V_i$ depending on whether the write of process~$j$ in~$M_1[j]$ was executed before or after the read of process~$i$ in $M_1[j]$ performed during the collect by process~$i$ of all the registers' content at round~1.  After round~2, the view $V_i$ of process~$i$ is a set of sets of inputs, whose content depends on the input values to the processes, and on the interleaving of the reads and writes performed by the processes during rounds~1 and~2. And so on, until process~$i$ obtains its final view~$V_i$ after~$t$ rounds. At this point, process~$i$ applies some function~$f$ on this view, which results in the output~$y_i$ decided by process~$i$.  The function~$f$ is specific of the algorithm, and distinct $t$-round algorithms differ solely in the function~$f$ applied to the view after $t$~rounds of communication. 

Interestingly, the generic form of round-based algorithms assumes that all internal computations are postponed to the very end of the execution. This is without loss of generality because, as the model is full-information, any algorithm in which internal computations are inserted between write and read operations can be simulated by an algorithm satisfying the generic form. A generic algorithm can tolerate an arbitrary large number (up to $n-1$) of crash failures, as any correct process terminates after $t$ rounds, independently from the number of processes that crash.  

\subsubsection{Snapshot and Immediate Snapshot}

Note that while the read operation is supposed to be atomic, the collect operation is not atomic, for it is just a sequence of $n$ successive read operations, one in each register. So, in particular, other processes can write in between two reads performed by a process, and, as mentioned before, different interleavings of reads and writes result in different views. 
Two classical stronger variants of the collect instruction are considered in this paper, including immediate snapshot used since \cite{BorowskyG93,SaksZ93}, namely: 
\begin{description}
	\item[Snapshot:] the collect operation itself is atomic, that is, all registers  are  read simultaneously at once, in an atomic manner.  
	\item[Immediate Snapshot:] the write-snapshot sequence of operations is itself atomic,  that is, not only all registers are read simultaneously at once, but the snapshot occurs ``immediately'' after the write operation, in an atomic manner. 
\end{description}

It is  easier to reason about algorithms using (immediate) snapshots than with algorithms using sequential collect as the number of interleavings between writes and reads are more restricted when using snapshots than when using collects, as it should appear clear in the next subsection. 

Regarding linearizability, note that while the atomic  read, write, and snapshot instructions may be supposed to occur at different times, two immediate snapshots, say by processes~$i$ and~$j$, may be executed concurrently. In this case, processes~$i$ (resp.,~$j$)  is supposed to read~$R[j]$ (resp., $R[i]$) after process~$j$ (resp.,~$i$) has written in it. 

\subsubsection{Topological transformations} 

Given a simplex~$\sigma=\{(i,x_i):i\in I\}\in \m{I}$ for some $I\subseteq [n]$,  one round of communication performed by the processes in~$\ID(\sigma)$ as in the generic round-based  Algorithm~\ref{alg:generic-model}
results in various possible simplices, depending on the interleaving of the different write and read operations. Such a simplex is of the form $\tau=\{(i,V_i):i\in I\}$, where $V_i=\{(j,x_j):j\in J_i\}$ is the view of process~$i$ after one round. This view, or, equivalently, the set~ $J_i\subseteq I$, depends on the communication model~$\Xi$, i.e., collect, snapshot, or immediate snapshot. These simplices induces a complex, denoted by~$\Xi_1(\sigma)$. Again, this complex has a specific form, which differs according to the three  models considered on this paper. 
To describe $\Xi_1(\sigma)$ in the basic case of collect,  we use the matrix representation of an execution from~\cite{2013TopologyComplexView_K,BenavidesR18}.
Let us consider the set $\collect(I)$ of matrices 
\[
M=\left [ \begin{array}{cccc}
	P_0 & P_1 & \dots & P_r \\
	I_0 & I_1 & \dots & I_r 
\end{array}\right ]
\]
such that (1)~$0\leq r \leq |I|-1$, (2)~$P_s\subseteq I$ for all $s=0,\dots,r$, (3)~$P_0=I$, (4)~$I_0,\dots,I_r$ form a partition of~$I$, and (5) for every $s=0,\dots,r$, $\cup_{j=s}^{r}I_j\subseteq P_s$. The semantic of such a matrix is that every process~$i\in I_s$ for $s\in \{0,\dots,r\}$ has read all input values from the processes in~$P_s$, and therefore its view after one round is $V_i=\{(j,x_j):j \in P_s\}$. The simplex in $\Xi_1(\sigma)$ corresponding to the communication pattern described by~ $M$ can therefore be written as 
$
\tau_M=\{(i,V_i): i\in I\} 
$
where, for every $s\in \{0,\dots,r\}$ and every $i\in I_s$, $V_i=\{(j,x_j):j \in P_s\}$. Given an input simplex $\sigma\in\m{I}$, the complex $\Xi_1(\sigma)$ can now be described as follows, depending on the communication model. 

\begin{description}
	\item[Topological transformation for collect:] $\Xi_1(\sigma)$ is the complex whose facets are all the  simplices $\tau_M$ such that $M\in\collect(\ID(\sigma))$. 
	\item[Topological transformation for snapshot:] $\Xi_1(\sigma)$ is the complex whose facets are all the  simplices $\tau_M$ such that $M\in\collect(\ID(\sigma))$ and, for every $0\leq i< j\leq r$, $P_i\subseteq P_j$ or $P_j\subseteq P_i$.\footnote{This is because, for any two processes $p$ and~$q$, either $p$ performs its snapshot before $q$, or after~$q$. In other words, two views must be comparable according the subset relation, and thus the views form a chain according to the subset order relation.}
	\item[Topological transformation for immediate snapshot:]  $\Xi_1(\sigma)$ is the complex whose facets are all the  simplices $\tau_M$ such that $M\in\collect(\ID(\sigma))$ and, for every $0\leq i \leq r$, every $p\in I_i$, and every $q\in P_i$, if $q\in I_j$ for some $j\in \{0,\dots,r\}$ then $P_j\subseteq P_i$.\footnote{This is because, as a snapshot occur ``immediately after'' its corresponding write,  if process~$p_1$ has process~$p_2$ in its snapshot, and $p_2$ has process~$p_3$ in its snapshot, then $p_1$ has $p_3$ in its snapshots  too --- while this is not necessarily the case with (non immediate) snapshots. }
\end{description} 
Figure~\ref{fig:3topologies} in Appendix~\ref{app:3topologies} displays an example of these three different complexes, for a 2-dimensional simplex~$\sigma$ (i.e., a system with three processes). 

\paragraph{Protocol Complex.}

The operator $\Xi_1$ can be iterated, yielding the sequence $(\Xi_t)_{t\geq 0}$ where, for every simplex~${\sigma\in\m{I}}$, $\Xi_0(\sigma)=\sigma$, and, for every $t\geq 1$, $\Xi_t(\sigma)=\Xi_1(\Xi_{t-1}(\sigma))$.  For every input complex~$\m{I}$, and every $t\geq 0$, the complex $\Xi_t(\m{I})$ is called the \emph{protocol complex} after $t$ rounds. This complex is the union, for all simplices $\sigma\in\m{I}$ of the complex $\Xi_t(\sigma)$ resulting from $t$~rounds starting from input~$\sigma$.

\medskip

Note that for any two input simplices~$\sigma=\{(i,x_i):i\in I\}$ and~$\sigma'=\{(i,x'_i):i\in I\}$, the complexes $\Xi_1(\sigma)$ and $\Xi_1(\sigma')$ are isomorphic. We denote by $\chi:\Xi_1(\sigma)\to\Xi_1(\sigma')$ the canonical isomorphism that maps vertex $v=(i,\{(j,x_j):j\in J_i\})$ of  $\Xi_1(\sigma)$ to vertex $\chi(v)={(i,\{(j,x'_j):j\in J_i\})}$ of $\Xi_1(\sigma')$.

\subsection{Task Solvability}

Let $\Xi$ be a round-based full-information model, and let $\Pi=(\m{I},\m{O},\Delta)$ be a task. Recall that a task $\Pi$ is solvable in $t$~rounds if there exists a simplicial map 
$
f:\Xi_t(\m{I})\to \m{O}
$
from the $t$-round protocol complex to the output complex that agrees with $\Delta$, i.e., for every $\sigma\in\m{I}$, 
$
f(\Xi_t(\sigma))\subseteq \Delta(\sigma).
$
Indeed, the simplicial map $f$ is merely the function $f$ used in Algorithm~\ref{alg:generic-model}
for computing the output values of the processes. The  former agrees with the input-output specification~$\Delta$ of the task if and only if the set of outputs computed according to the latter is correct w.r.t. the given input.

\newpage 
\section{The 1-Round Protocol Complex}
\label{app:3topologies}

\begin{figure}[h]
	\centering
	\includegraphics[width=6cm]{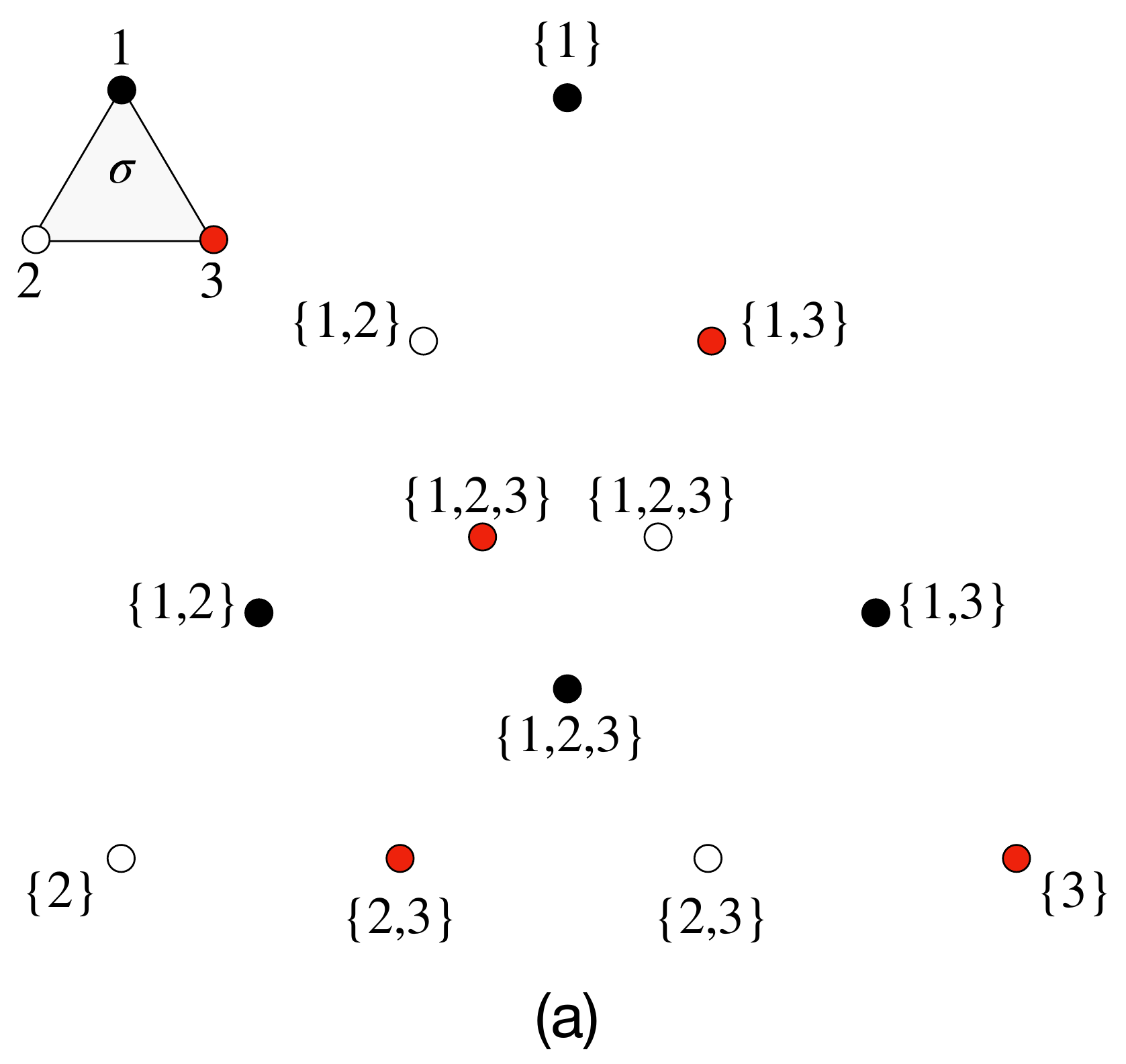}\hspace{1cm} \includegraphics[width=6cm]{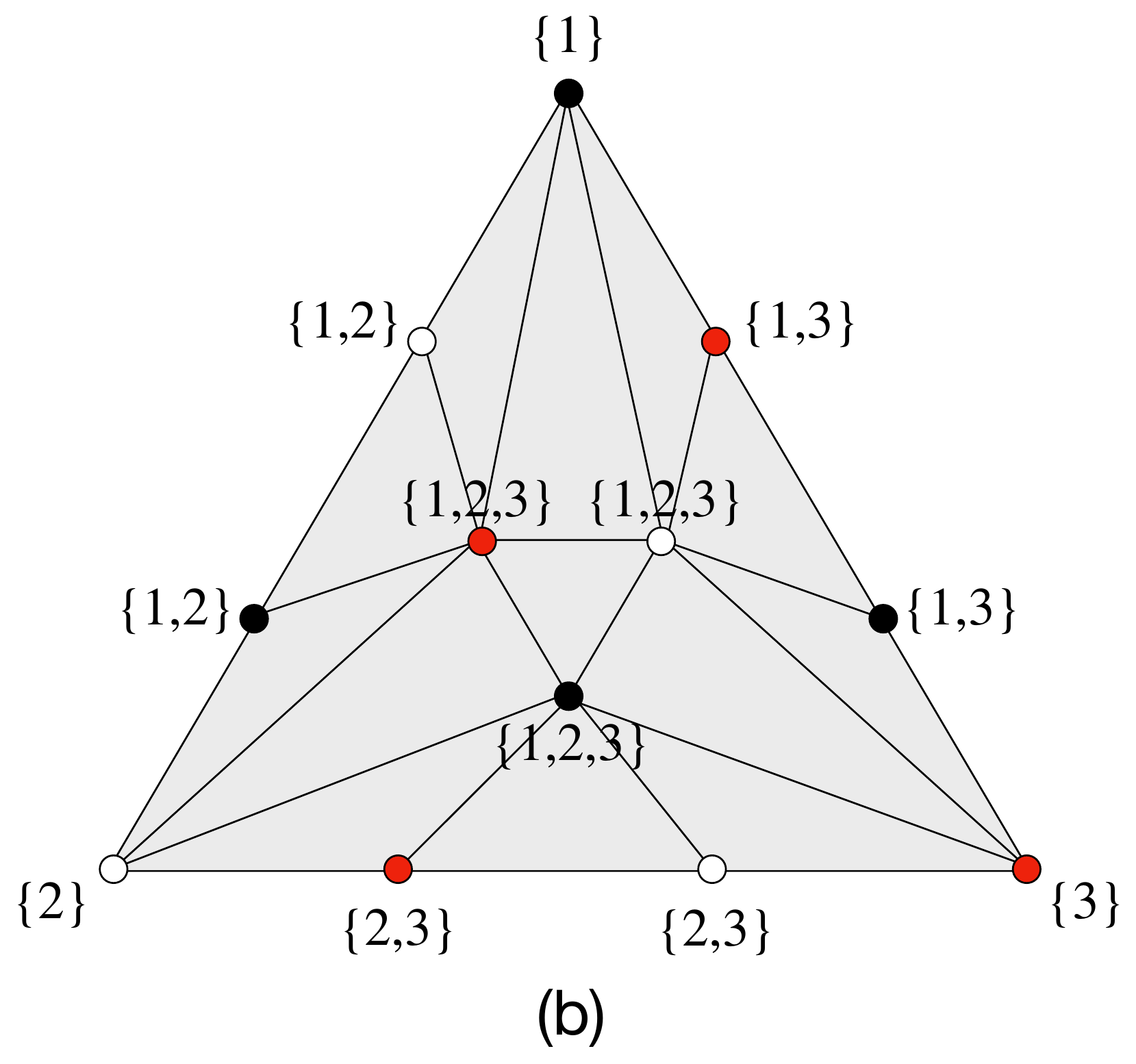}
	
	\vspace{.8cm}
	
	\includegraphics[width=6cm]{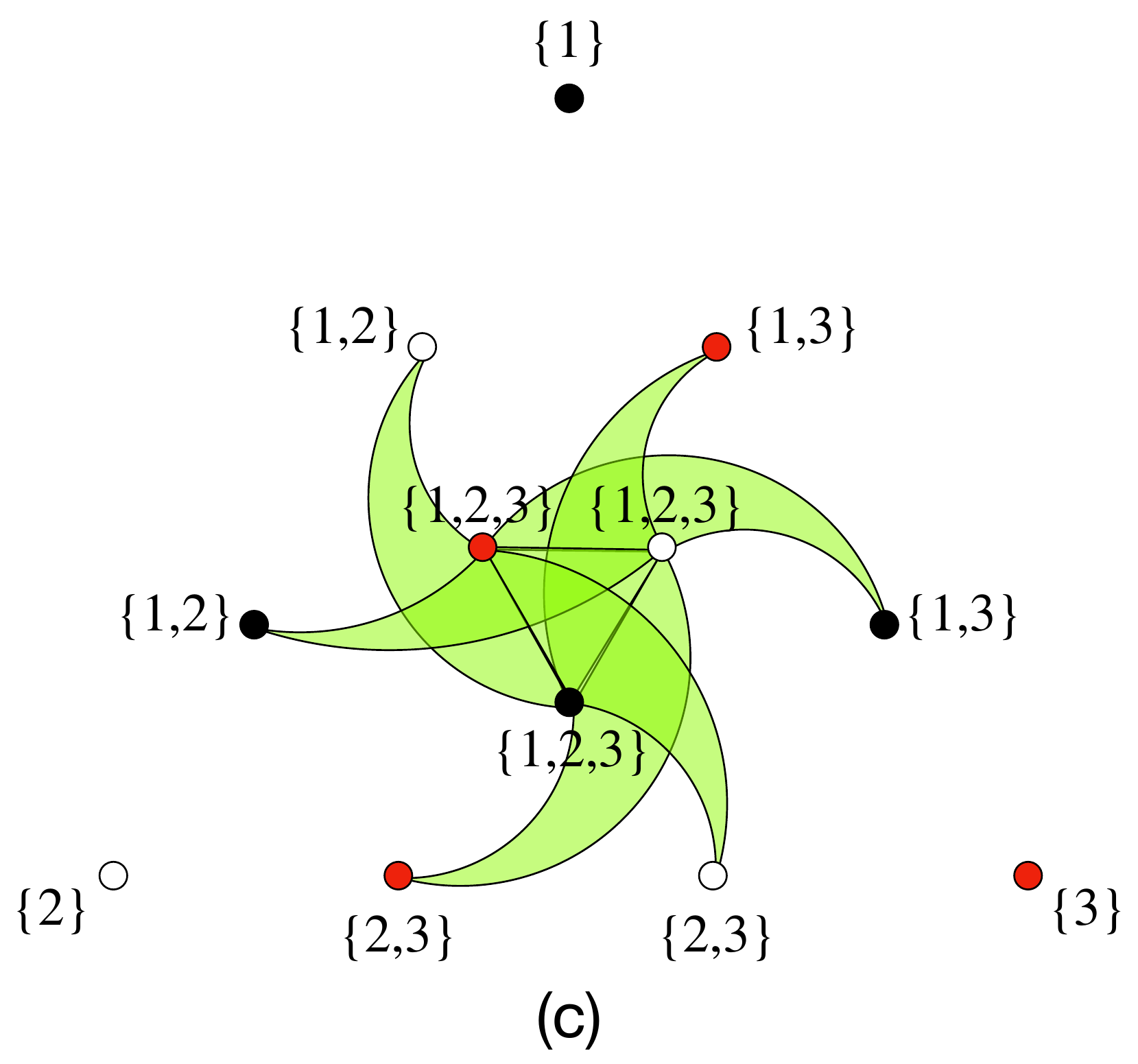} \hspace{1cm} \includegraphics[width=6cm]{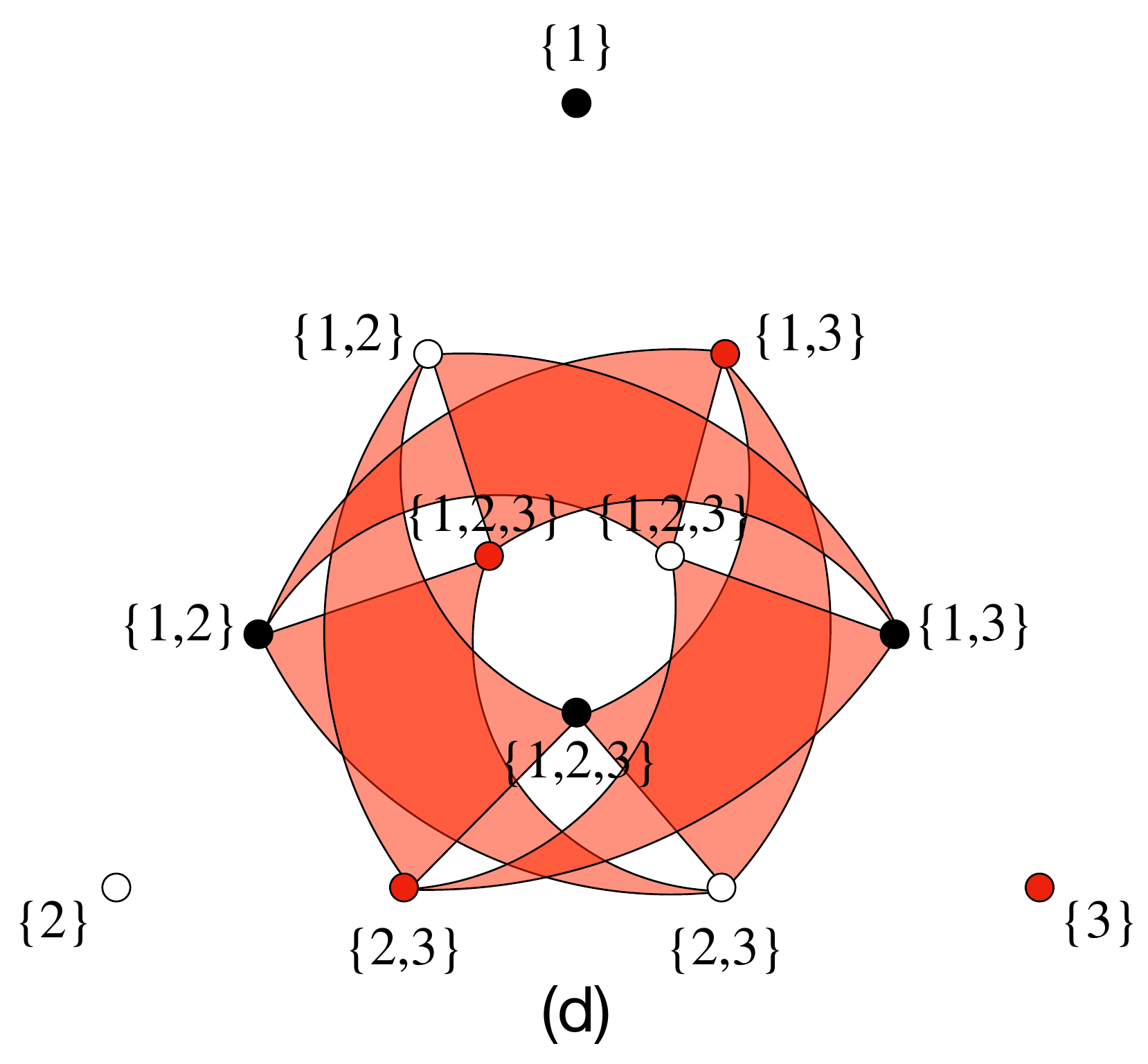}
	
	\caption{\sl 
		{\rm\sf (a)} All possible views after one round starting from the input simplex~$\sigma$ containing three processes (black, white, and red) with respective inputs~1, 2, and~3.  
		{\rm\sf (b)} All simplices in $\Xi_1(\sigma)$ resulting from immediate snapshot --- they result from a chromatic subdivision of~$\sigma$. 
		{\rm\sf (c)} All simplices in $\Xi_1(\sigma)$ resulting from snapshot but not from  immediate snapshot.
		{\rm\sf (d)} All simplices in $\Xi_1(\sigma)$ resulting from collect, but neither from snapshot nor immediate snapshot. 
	}
	\label{fig:3topologies}
\end{figure}

\vfill

\end{document}